\documentclass[runningheads]{llncs}

\usepackage{amsmath,amssymb}
\usepackage{thm-restate}
\usepackage{stmaryrd,mathrsfs,galois}
\usepackage{subcaption}

\usepackage[color]{changebar}
\cbcolor{red}

\usepackage[normalem]{ulem}

\AtBeginDocument{%
\setlength{\textfloatsep}{0.4\textfloatsep}
\setlength{\floatsep}{0.4\floatsep}
\setlength{\intextsep}{0.4\intextsep}
\setlength{\intextsep}{0.4\intextsep}
\setlength{\abovecaptionskip}{0.4\abovecaptionskip}
\setlength{\belowcaptionskip}{0.4\belowcaptionskip}
\setlength{\abovedisplayskip}{0.4\abovedisplayskip}
\setlength{\belowdisplayskip}{0.4\belowdisplayskip}
\clubpenalty=150
\widowpenalty=0
\displaywidowpenalty=50
}

\newcommand{\romanqed}{$\righthalfcup$}
\renewenvironment{example}[1][]{
  \par
  \trivlist
\item\ignorespaces
  \refstepcounter{example}
  \textit{Example \arabic{example}.}\ifx&#1& \else \;(#1) \fi \hspace{1pt}
}{
  \null\hfill\romanqed\endtrivlist
}

\newcommand{\tuple}[1]{\langle #1 \rangle}
\newcommand{\defeq}{\triangleq}
\renewcommand{\vec}[1]{\mathbf{#1}}

\newcommand{\hadamard}[2]{#1 * #2}
\newcommand{\proj}[1]{\Pi_{#1}} 
\newcommand{\rim}[2]{#1 \boxtimes #2} 
\newcommand{\image}[2]{\textit{image}(#1,#2)}

\newcommand{\LRA}{\exists\text{LRA}}
\newcommand{\LIRA}{\exists\text{LIRA}}

\newcommand{\tfsem}[1]{\mathbf{TF}\llbracket #1 \rrbracket}
\newcommand{\ite}[3]{\textbf{if } #1 \textbf{ then } #2 \textbf{ else } #3}
\newcommand{\while}[2]{\textbf{while } #1 \textbf{ do } #2}
\newcommand{\VAS}{$\mathbb{Q}$-VASR}
\newcommand{\VASS}{$\mathbb{Q}$-VASRS}

\newcommand{\VASSYM}{V}
\newcommand{\SYMVASS}{\mathcal{V}}
\newcommand{\res}{\vec{r}}
\newcommand{\resnv}{r}
\newcommand{\adv}{\vec{a}}

\newcommand{\reach}[1]{\textit{reach}(#1)}

\DeclareFontFamily{U}  {MnSymbolC}{}
\DeclareFontShape{U}{MnSymbolC}{m}{n}{
    <-6>  MnSymbolC5
   <6-7>  MnSymbolC6
   <7-8>  MnSymbolC7
   <8-9>  MnSymbolC8
   <9-10> MnSymbolC9
  <10-12> MnSymbolC10
  <12->   MnSymbolC12}{}
\DeclareSymbolFont{MnSyC}{U}{MnSymbolC}{m}{n}
\DeclareMathSymbol{\righthalfcup}{\mathrel}{MnSyC}{184}
\DeclareMathSymbol{\ostar}{\mathrel}{MnSyC}{102}

\usepackage[ruled,linesnumbered,vlined,commentsnumbered]{algorithm2e}
\usepackage{tikz-cd}
\usetikzlibrary{shapes}

\begin{document}

\title{Loop Summarization with Rational Vector Addition Systems}
\author{Jake Silverman \and Zachary Kincaid}
\institute{Princeton University}
\maketitle

\begin{abstract}
  This paper presents a technique for computing numerical loop
  summaries.  The method synthesizes a rational vector addition system
  with resets (\VAS{}) that simulates the action of an input loop, and
  then uses the reachability relation of that \VAS{} to over-approximate
  the behavior of the loop.  The key technical problem solved in this
  paper is to automatically synthesize a \VAS{} that is a \textit{best
    abstraction} of a given loop in the sense that (1) it simulates
  the loop and (2) it is simulated by any other \VAS{} that simulates
  the loop.  Since our loop summarization scheme is based on computing
  the \textit{exact} reachability relation of a \textit{best}
  abstraction of a loop, we can make theoretical guarantees about its
  behavior.  Moreover, we show experimentally that the technique is precise and
  performant in practice.
\end{abstract}

\section{Introduction} \label{sec:introduction}
Modern software verification techniques employ a number of heuristics
for reasoning about loops.  While these heuristics are often
effective, they are unpredictable.  For example, an abstract
interpreter may fail to find the most precise invariant expressible in
the language of its abstract domain due to imprecise widening, or a
software-model checker might fail to terminate because it generates
interpolants that are insufficiently general.  
This paper presents a loop summarization technique that is capable of 
generating loop invariants in an expressive and decidable language and 
provides theoretical guarantees about invariant quality.

The key idea behind our technique is to leverage reachability
results of vector addition systems (VAS) for invariant generation.
Vector addition systems are a class of infinite-state transition
systems with decidable reachability, classically used as a
model of parallel systems \cite{KM1969}.  We consider a variation
of VAS, \textit{rational VAS with resets (\VAS{})}, wherein there is a finite
number of rational-typed variables and a finite set of transitions
that simultaneously update each variable in the system by either
adding a constant value or (re)setting the variable to a constant
value.  Our interest in \VAS{}s stems from the fact that there is
(polytime) procedure to compute a linear arithmetic formula that
represents a \VAS's reachability relation \cite{RP:HH2014}.

Since the reachability relation of a \VAS{} is computable, the
dynamics of \VAS{} can be analyzed without relying on heuristic
techniques.  However, there is a gap between \VAS{} and the loops that
we are interested in summarizing. The latter typically use a rich set of
operations (memory manipulation, conditionals, non-constant
increments, non-linear arithmetic, etc) and cannot be analyzed
precisely. We bridge the gap with a procedure that, for any loop,
synthesizes a \VAS{} that simulates it.
The reachability relation of
the \VAS{} can then be used to over-approximate the behavior of the loop.
Moreover, we prove that if a loop is expressed in linear rational
arithmetic (LRA), then our procedure synthesizes a \textit{best} \VAS{}
abstraction, in the sense that it simulates any other \VAS{} that
simulates the loop.    That is, imprecision in the analysis is due to 
inherent limitations of the \VAS{} model, rather heuristic algorithmic choices.

One limitation of the model is that \VAS{}s over-approximate
multi-path loops by treating the choice between paths as
non-deterministic.  We show that \VASS{}, \VAS{} extended with control
states, can be used to improve our invariant generation scheme by
encoding control flow information and inter-path control dependencies
that are lost in the \VAS{} abstraction.  We give an algorithm for
synthesizing a \VASS{} abstraction of a given loop, which (like our
\VAS{} abstraction algorithm) synthesizes \textit{best} abstractions
under certain assumptions.

Finally, we note that our analysis techniques extend to complex
control structures (such as nested loops) by employing summarization
compositionally (i.e., ``bottom-up'').  For example, our analysis
summarizes a nested loop by first summarizing its inner loops, and
then uses the summaries to analyze the outer loop.  As a result of
compositionality, our analysis can be applied to partial programs, is
easy to parallelize, and has the potential to scale to large code
bases.

The main contributions of the paper are as follows:
\begin{itemize}
\item We present a procedure to synthesize \VAS{} abstractions of
  transition formulas.  For transition formulas in linear rational
  arithmetic, the synthesized \VAS{} abstraction is a \textit{best} abstraction.
\item We present a technique for improving the precision of our
  analysis by using \VAS{} with states to capture loop control structure.
\item We implement the proposed loop summarization techniques and show
  that their ability to verify user assertions is comparable to
  software model checkers, while at the same time providing
  theoretical guarantees of termination and invariant quality.
\end{itemize}

\subsection{Outline}\label{subsec:outline}

This section illustrates the high-level structure of our invariant
generation scheme.  The goal is to compute a \textit{transition
  formula} that summarizes the behavior of a given program.  A
transition formula is a formula over a set of program variables
$\textsf{Var}$ along with primed copies $\textsf{Var}'$, representing
the state of the program before and after executing a computation
(respectively).  For any given program $P$, a transition formula
$\tfsem{P}$ can be computed by recursion on syntax:\footnote{This
  style of analysis can be extended from a simple block-structured
  language to one with control flow and recursive procedures using the
  framework of algebraic program analysis
  \cite{JACM:Tarjan1981b,PLDI:KBFR17}.}
\begin{align*}
  \tfsem{\texttt{x := }e} &\defeq \text{x}' = e \land \bigwedge_{\texttt{y}\neq\texttt{x} \in \textsf{Var}} \texttt{y}' = \texttt{y}\\
  \tfsem{\ite{c}{P_1}{P_2}} &\defeq (c \land \tfsem{P_1}) \lor (\lnot c \land \tfsem{P_2})\\
  \tfsem{P_1 \texttt{;} P_2} & \defeq \exists X \in \mathbb{Z}. \tfsem{P_1}[\textsf{Var}' \mapsto X] \land \tfsem{P_2}[\textsf{Var} \mapsto X]\\
  \tfsem{\while{c}{P}} &\defeq (c \land \tfsem{P})^\star \land (\lnot c[\textsf{Var} \mapsto \textsf{Var}'])
\end{align*}
where $(-)^\star$ is a function that computes an over-approximation of
the transitive closure of a transition formula.  The contribution of
this paper is a method for computing this $(-)^\star$ operation, which
is based on first over-approximating the input transition formula by a
\VAS{}, and then computing the (exact) reachability relation of the \VAS{}.

\begin{figure}[t]
\begin{subfigure}[t]{6.5cm}
  \textbf{procedure} \texttt{enqueue(elt):}\\
  \hspace*{10pt}\texttt{back := cons(elt,back)}\\
  \hspace*{10pt}\texttt{size := size + 1}\\
  \vspace*{2pt}\\
  \textbf{procedure} \texttt{dequeue():}\\
  \hspace*{10pt}\textbf{if} \texttt{(front == nil)} \textbf{then}\\
  \hspace*{20pt}{\color{gray}\texttt{// Reverse back, append to front}}\\
  \hspace*{20pt}\textbf{while} \texttt{(back != nil)} \textbf{do}\\
  \hspace*{30pt}\texttt{front := cons(head(back),front)}\\
  \hspace*{30pt}\texttt{back := tail(back)}\\
  \hspace*{10pt}\texttt{result := head(front)}\\
  \hspace*{10pt}\texttt{front := tail(front)}\\
  \hspace*{10pt}\texttt{size := size - 1}\\
  \hspace*{10pt}\textbf{return} \texttt{result}
  \subcaption{\label{fig:mem-queue}Persistent queue}
\end{subfigure}
\begin{subfigure}[t]{6.5cm}
  \textbf{procedure} \texttt{enqueue():}\\
  \hspace*{10pt}\texttt{back\_len := back\_len + 1}\\
  \hspace*{10pt}\texttt{mem\_ops := mem\_ops + 1}\\
  \hspace*{10pt}\texttt{size := size + 1}\\
  \textbf{procedure} \texttt{dequeue():}\\
  \hspace*{10pt}\textbf{if} \texttt{(front\_len == 0)} \textbf{then}\\
  \hspace*{20pt}\textbf{while} \texttt{(back\_len != 0)} \textbf{do}\\
  \hspace*{30pt}\texttt{front\_len := front\_len + 1}\\
  \hspace*{30pt}\texttt{back\_len := back\_len - 1}\\
  \hspace*{30pt}\texttt{mem\_ops := mem\_ops + 3}\\
  \hspace*{10pt}\texttt{size := size - 1}\\
  \hspace*{10pt}\texttt{front\_len := front\_len - 1}\\
  \hspace*{10pt}\texttt{mem\_ops := mem\_ops + 2}\\
  \textbf{procedure} \texttt{harness():}\\
  \hspace*{10pt}\texttt{nb\_ops := 0}\\
  \hspace*{10pt}\textbf{while} \texttt{nondet()} \textbf{do}\\
  \hspace*{20pt}{nb\_ops := nb\_ops + 1}\\
  \hspace*{20pt}\textbf{if} \texttt{(size > 0 \&\& nondet())}\\
  \hspace*{30pt}enqueue()\\
  \hspace*{20pt}\textbf{else}\\
  \hspace*{30pt}dequeue()
  \subcaption{\label{fig:int-queue}Integer model \& harness}
\end{subfigure}
\caption{\label{fig:queue}A persistent queue and integer model.
  \texttt{back\_len} and \texttt{front\_len} models the lengths of the
  lists \texttt{front} and \texttt{back}; \texttt{mem\_ops} counts the
  number of memory operations in the computation.}
\end{figure}

We illustrate the analysis on an integer model
of a persistent queue data structure, pictured in
Figure~\ref{fig:queue}.  The example consists of two operations
(\texttt{enqueue} and \texttt{dequeue}), as well as a test harness
(\texttt{harness}) that non-deterministically executes \texttt{enqueue}
and \texttt{dequeue} operations.  The queue achieves $O(1)$ amortized
memory operations (\texttt{mem\_ops}) in \texttt{enqueue} and \texttt{queue} by implementing the queue as
two lists, \texttt{front} and \texttt{back} (whose lengths are modeled
as \texttt{front\_len} and \texttt{back\_len}, respectively): 
the sequence of elements in the queue is the \texttt{front}
list followed by the reverse of the \texttt{back} list.  We will show
that the queue functions use $O(1)$ amortized memory operations by finding a summary
for \texttt{harness} that implies a linear bound on \texttt{mem\_ops}
(the number of memory operations in the computation) in terms of
\texttt{nb\_ops} (the total number of
\texttt{enqueue}/\texttt{dequeue} operations executed in some
sequence of operations).

We analyze the queue compositionally, in ``bottom-up''
fashion (i.e., starting from deeply-nested code and working our way
back up to a summary for \texttt{harness}).  There are two loops of
interest, one in \texttt{dequeue} and one in \texttt{harness}.  Since
the \texttt{dequeue} loop is nested inside the \texttt{harness} loop,
\texttt{dequeue} is analyzed first.  We start by computing a transition
formula that represents one execution of the body of the \texttt{dequeue} loop:
\[
\textit{Body}_{\texttt{deq}} =
\texttt{back\_len} > 0 \land
\left(\begin{array}{@{\hspace{0pt}}l@{\hspace{1pt}}l@{\hspace{0pt}}}
  & \texttt{front\_len}' = \texttt{front\_len} + 1\\
  \land & \texttt{back\_len}' = \texttt{back\_len} - 1\\
  \land & \texttt{mem\_ops}' = \texttt{mem\_ops} + 3\\
  \land & \texttt{size}' = \texttt{size}
\end{array}\right)
\]
Observe that each variable in the loop is incremented by a constant
value.  As a result, the loop update can be captured faithfully by a
vector addition system.  In particular, we see that this loop body
formula is simulated by the \VAS{} $V_{\texttt{deq}}$ (below), where the
correspondence between the state-space of
$\textit{Body}_{\texttt{deq}}$ and $V_{\texttt{deq}}$ is given by the
identity transformation (i.e., each dimension of $V_{\texttt{deq}}$
simply represents one of the variables of
$\textit{Body}_{\texttt{deq}}$).
{\small
  \[\begin{bmatrix}w\\x\\y\\z\end{bmatrix} = \begin{bmatrix}1&0&0&0\\0&1&0&0\\0&0&1&0\\0&0&0&1\end{bmatrix}\begin{bmatrix}\texttt{front\_len}\\\texttt{back\_len}\\\texttt{mem\_ops}\\\texttt{size}\\\end{bmatrix};
  \hspace*{0.3cm}
    V_{\texttt{deq}} = \left\{
    \begin{bmatrix}w \\ x \\ y \\ z\end{bmatrix}
      \rightarrow
      \begin{bmatrix}w + 1 \\ x-1 \\ y + 3 \\ z\end{bmatrix}
    \right\}\ .
    \]}
A formula representing the reachability relation of a vector
    addition system can be computed in polytime.  For the case of
    $V_{\texttt{deq}}$, a formula representing $k$ steps of the \VAS{} is
    simply
    \begin{equation} \label{eq:vas-reach}
      w' = w + k \land x' = x - k \land y' = y + 3k \land z' = z\ .\tag{$\dagger$}
    \end{equation}

    To capture information about the pre-condition of the loop, we can
    project the primed variables to obtain $\texttt{back\_len} > 0$;
    similarly, for the post-condition, we can project the unprimed
    variables to obtain $\texttt{back\_len}' \geq 0$.  Finally,
    combining ($\dagger$) (translated back into the
    vocabulary of the program) and the pre/post-condition, we form the
    following approximation of the \texttt{dequeue} loop's behavior:
  \[ \exists k. k \geq 0 \land \left(\begin{array}{@{\hspace{0pt}}l@{\hspace{1pt}}l@{\hspace{0pt}}}
  & \texttt{front\_len}' = \texttt{front\_len} + k\\
  \land & \texttt{back\_len}' = \texttt{back\_len} - k\\
  \land & \texttt{mem\_ops}' = \texttt{mem\_ops} + 3k\\
  \land & \texttt{size}' = \texttt{size}
\end{array}\right) \land \left(k > 0 \Rightarrow \left(\begin{array}{@{\hspace{0pt}}l@{\hspace{1pt}}l@{\hspace{0pt}}}& \texttt{back\_len} > 0\\ \land & \texttt{back\_len}' \geq 0)\end{array}\right)\right)\ .
\]

  Using this summary for the \texttt{dequeue} loop, we proceed to
  compute a transition formula for the body of the \texttt{harness}
  loop (omitted for brevity).  Just as with the \texttt{dequeue} loop,
  we analyze the \texttt{harness} loop by synthesizing a \VAS{}
  that simulates it, $V_{\texttt{har}}$ (below), where the correspondence between the state space of the \texttt{harness} loop and $V_{\texttt{har}}$ is given by the transformation $S_{\texttt{har}}$:

  $\begin{array}{l}
    \begin{bmatrix}v\\w\\x\\y\\z\end{bmatrix} =
  \underbrace{\begin{bmatrix}
    0 & 0 & 0 & 1 & 0\\
    0 & 1 & 0 & 0 & 0\\
    0 & 3 & 1 & 0 & 0\\
    1 & 1 & 0 & 0 & 0\\
    0 & 0 & 0 & 0 & 1\\
  \end{bmatrix}}_{S_{\texttt{har}}}
  \begin{bmatrix}\texttt{front\_len}\\\texttt{back\_len}\\\texttt{mem\_ops}\\\texttt{size}\\\texttt{nb\_ops}\end{bmatrix}; i.e., \left(\begin{array}{ll}
      &\texttt{size} = v\\
      \land&\texttt{back\_len} = w\\
      \land&\texttt{mem\_ops}+3\texttt{back\_len} = x\\
      \land&\texttt{back\_len} + \texttt{front\_len} = y\\
      \land&\texttt{nb\_ops} = z\\
    \end{array}\right).
  \\
    V_{\texttt{har}} = \left\{\underbrace{\begin{bmatrix}v\\w\\x\\y\\z\end{bmatrix} \rightarrow \begin{bmatrix}v+1\\w+1\\x+4\\y+1\\z+1\end{bmatrix}}_{\texttt{enqueue}},
\underbrace{\begin{bmatrix}v\\w\\x\\y\\z\end{bmatrix} \rightarrow \begin{bmatrix}v-1\\w\\x+2\\y-1\\z+1\end{bmatrix}}_{\texttt{dequeue} \text{ fast}},
\underbrace{\begin{bmatrix}v\\w\\x\\y\\z\end{bmatrix} \rightarrow \begin{bmatrix}v-1\\0\\x+2\\y-1\\z+1\end{bmatrix}}_{\texttt{dequeue} \text{ slow}}
\right\}
  \end{array}$

    Unlike the \texttt{dequeue} loop, we do not get an exact
    characterization of the dynamics of each changed variable.  In
    particular, in the slow \texttt{dequeue} path through the loop,
    the value of \texttt{front\_len}, \texttt{back\_len}, and
    \texttt{mem\_ops} change by a variable amount.  Since
    \texttt{back\_len} is set to 0, its behavior can be captured by
    a reset.  The dynamics of \texttt{front\_len} and
    \texttt{mem\_ops} cannot be captured by a \VAS{}, but (using our
    \texttt{dequeue} summary) we can observe that the sum of
      \texttt{front\_len} + \texttt{back\_len} is decremented by 1,
      and the sum of \texttt{mem\_ops} + 3\texttt{back\_len} is incremented
      by 2.

    We compute the following formula that captures the reachability
    relation of $V_{\texttt{har}}$ (taking $k_1$ steps of
    \texttt{enqueue}, $k_2$ steps of \texttt{dequeue} fast, and $k_3$
    steps of \texttt{dequeue} slow) under the inverse image of the
    state correspondence $S_{\texttt{har}}$:
    \[
    \left(\begin{array}{@{\hspace{0pt}}l@{\hspace{1pt}}l@{\hspace{0pt}}}
      & \texttt{size}' = \texttt{size} + k_1 - k_2 - k_3\\
      \land & \left((k_3 = 0 \land \texttt{back\_len}' = \texttt{back\_len} + k_1) \lor (k_3 > 0 \land 0 \leq \texttt{back\_len}' \leq k_1)\right)\\
      \land &\texttt{mem\_ops}' + 3\texttt{back\_len}' = \texttt{mem\_ops} + 3\texttt{back\_len} + 4k_1 + 2k_2 + 2k_3\\
      \land &\texttt{front\_len}' + \texttt{back\_len}' = \texttt{front\_len} + \texttt{back\_len} + k_1 - k_2 - k_3\\
      \land &\texttt{nb\_ops}' = \texttt{nb\_ops} + k_1 + k_2 + k_3
    \end{array}\right)
   \]
    From the above formula (along with pre/post-condition
    formulas), we obtain a summary for the \texttt{harness} loop
    (omitted for brevity).  Using this summary we can prove 
    (supposing that we
    start in a state where all variables are zero) that
    \texttt{mem\_ops}
    is at most 4 times \texttt{nb\_ops} (i.e., \texttt{enqueue} and
    \texttt{dequeue} use O(1) amortized memory operations).



\section{Background} \label{sec:background}
The syntax of $\LIRA{}$, the existential fragment of linear
integer/rational arithmetic, is given by the following grammar:
\begin{align*}
  s,t \in \textsf{Term} &::= c \mid x \mid s + t \mid c \cdot t\\
  F,G \in \textsf{Formula} &::= s < t \mid s = t \mid F \land G \mid F \lor G \mid \exists x \in \mathbb{Q}. F \mid \exists x \in \mathbb{Z}. F
\end{align*}
where $x$ is a (rational sorted) variable symbol and $c$ is a rational
constant.  Observe that (without loss of generality) formulas are free
of negation.  $\LRA{}$ (linear rational arithmetic) refers to the
fragment of $\LIRA{}$ that omits quantification over the integer sort.

A \textbf{transition system} is a pair $(S, \rightarrow)$ where $S$ is
a (potentially infinite) set of states and $\rightarrow \subseteq S
\times S$ is a transition relation.  For a transition relation
$\rightarrow$, we use $\rightarrow^*$ to denote its reflexive,
transitive closure.

A \textbf{transition formula} is a formula $F(\vec{x},\vec{x}')$ whose
free variables range over $\vec{x} = x_1,...,x_n$ and $\vec{x}' =
x_1',...,x_n'$ (we refer to the number $n$ as the \textit{dimension}
of $F$); these variables designate the state before and after a
transition.  In the following, we assume that transition formulas are
defined over $\LIRA{}$.  For a transition formula
$F(\vec{x},\vec{x}')$ and vectors of terms $\vec{s}$ and $\vec{t}$, we
use $F(\vec{s},\vec{t})$ to denote the formula $F$ with each $x_i$
replaced by $s_i$ and each $x_i'$ replaced by $t_i$.  A transition
formula $F(\mathbf{x},\mathbf{x}')$ defines a transition system
$(S_{F}, \rightarrow_{F})$, where the state space $S_{F}$ is
$\mathbb{Q}^{n}$ and which can transition $\vec{u} \rightarrow_{F}
\vec{v}$ iff $F(\vec{u},\vec{v})$ is valid.


For two rational vectors $\vec{a}$ and $\vec{b}$ of the same dimension
$d$, we use $\vec{a} \cdot \vec{b}$ to denote the inner product
$\vec{a} \cdot \vec{b} = \sum_{i=1}^d a_ib_i$ and $\vec{a} * \vec{b}$
to denote the pointwise (aka Hadamard) product $(\vec{a} * \vec{b})_i
= a_i b_i$.  For any natural number $i$, we use $\vec{e}_i$ to denote
the standard basis vector in the $i$th direction (i.e., the vector consisting
of all zeros except the $i$th entry, which is 1), where the dimension
of $\vec{e}_i$ is understood from context.  We use $I_n$ to denote the
$n \times n$ identity matrix.




\begin{definition} A \textbf{rational vector addition system with resets} (\textbf{\VAS{}}) of dimension $d$ is a finite set $\VASSYM{} \subseteq \{0,1\}^d \times \mathbb{Q}^d$ of transformers.  Each transformer $(\res,\adv) \in \VASSYM{}$ consists of a binary \textit{reset vector} $\res$, and a rational \textit{addition vector}
  $\adv$, both of dimension $d$.  $\VASSYM{}$ defines a transition
  system $(S_{\VASSYM{}}, \rightarrow_{\VASSYM{}})$, where the state
  space $S_{\VASSYM{}}$ is $\mathbb{Q}^d$ and which can transition
  $\vec{u} \rightarrow_{\VASSYM{}} \vec{v}$ iff $\vec{v} = \res *
  \vec{u} + \adv$ for some $(\res,\adv) \in \VASSYM{}$.
\end{definition}


\begin{definition}
  A \textbf{rational vector addition system with resets and states} (\VASS{}) of dimension $d$ is a pair $\SYMVASS = (Q, E)$, where $Q$ is a finite set of control states, and $E \subseteq Q \times \{0,1\}^d \times \mathbb{Q}^d \times Q$ is a finite set of edges labeled by ($d$-dimensional) transformers.  $\SYMVASS$ defines a transition system $(S_{\SYMVASS}, \rightarrow_{\SYMVASS})$, where
the state space $S_{\SYMVASS}$ is $Q \times \mathbb{Q}^n$ and which can transition
$(q_1, \vec{u}) \rightarrow_{\SYMVASS} (q_2, \vec{v})$ iff
there is some edge $(q_1, (\res,\adv), q_2) \in E$ such that $\vec{v} = \res * \vec{u} + \adv$.
\end{definition}

Our invariant generation scheme is based on the following result, which is a simple consequence of the work of Haase and Halfon:
\begin{theorem}[\cite{RP:HH2014}] \label{thm:vass-reach}
  There is a polytime algorithm which, given a $d$-dimensional \VASS{} $\SYMVASS{} = (Q, E)$, computes an $\LIRA$ transition formula $\reach{\SYMVASS{}}$ such
  that for all $\vec{u},\vec{v} \in \mathbb{Q}^d$, we have
  $(p, \vec{u}) \rightarrow_{\SYMVASS{}}^* (q, \vec{v})$ for some control states $p, q \in Q$ if and only if
  $\vec{u} \rightarrow_{\reach{\SYMVASS{}}} \vec{v}$.
\end{theorem}

Note that \VAS{} can be realized as \VASS{} with a single control
state, so this theorem also applies to \VAS{}.


\section{Approximating loops with vector addition systems} \label{sec:vas-abstraction}

In this section, we describe a method for over-approximating the
transitive closure of a transition formula using a \VAS{}.  This
procedure immediately extends to computing summaries for programs
(including programs with nested loops)
using the method outlined in Section~\ref{subsec:outline}.

The core algorithmic problem that we answer in this section is:
\textit{given a transition formula, how can we synthesize a (best)
  abstraction of that formula's dynamics as a \VAS{}?} 
We begin by
formalizing the problem: in particular, we define what it means for a \VAS{} to
simulate a transition formula and what it means for an abstraction to
be ``best.''
\begin{definition}
  Let $A = (\mathbb{Q}^n,\rightarrow_A)$ and $B =
  (\mathbb{Q}^m,\rightarrow_B)$ be transition systems operating over
  rational vector spaces.  A \textbf{linear simulation} from $A$ to
  $B$ is a linear transformation $S : \mathbb{Q}^{m \times n}
  $ such that for all $\vec{u},\vec{v} \in \mathbb{Q}^n$
  for which $\vec{u} \rightarrow_A \vec{v}$, we have $S\vec{u} \rightarrow_B
  S\vec{v}$.  We use $A \Vdash_S B$ to denote that $S$ is a linear
  simulation from $A$ to $B$.
\end{definition}

Suppose that $F(\vec{x},\vec{x}')$ is an $n$-dimensional transition
formula, $\VASSYM{}$ is a $d$-dimensional \VAS{}, and $S : \mathbb{Q}^{d  \times n}$
is linear transformation.
The key property of simulations that underlies our loop summarization
scheme is that if $F \Vdash_S \VASSYM{}$, then
$\reach{\VASSYM{}}(S\mathbf{x},S\mathbf{x}')$ (i.e., the reachability
relation of $\VASSYM{}$ under the inverse image of $S$)
over-approximates the transitive closure of $F$.  Finally, we observe
that simulation $F \Vdash_S \VASSYM{}$ can equivalently be defined by
the validity of the entailment $F \models \gamma(S,\VASSYM{})$, where
\[
  \gamma(S,\VASSYM{}) \defeq \bigvee_{(\res, \adv) \in \VASSYM{}} S\vec{x}' = \hadamard{\res}{S\vec{x}} + \adv
  \]
is a transition formula that represents the transitions that $\VASSYM{}$
simulates under transformation $S$.


Our task is to synthesize a linear transformation $S$ and a \VAS{}
$\VASSYM{}$ such that ${F \Vdash_S \VASSYM{}}$.  We call a pair
$(S,\VASSYM{})$, consisting of a rational matrix $S \in \mathbb{Q}^{d
  \times n}$ and a $d$-dimensional \VAS{} $\VASSYM{}$, a
\textbf{\VAS{} abstraction}. We say that $n$ is the \textit{concrete
  dimension} of $(S,\VASSYM{})$ and $d$ is the \textit{abstract
  dimension}.  If $F \Vdash_S \VASSYM{}$, then we say that
$(S,\VASSYM{})$ is a \textbf{\VAS{} abstraction of $F$}. A transition
formula may have many \VAS{} abstractions; we are interested in
computing a \VAS{} abstraction $(S,\VASSYM{})$ that results in the
most precise over-approximation of the transitive closure of $F$.
Towards this end, we define a preorder $\preceq$ on \VAS{}
abstractions, where $(S^1,\VASSYM{}^1) \preceq (S^2,\VASSYM{}^2)$ iff
there exists a linear transformation $T \in \mathbb{Q}^{e\times d}$
such that $\VASSYM{}^1 \Vdash_T \VASSYM{}^2$ and $TS^1 = S^2$ (where
$d$ and $e$ are the abstract dimensions of $(S^1,\VASSYM{}^1)$ and
$(S^2,\VASSYM{}^2)$, respectively).  Observe that if
$(S^1,\VASSYM{}^1) \preceq (S^2,\VASSYM{}^2)$, then
$\reach{\VASSYM{}^1}(S^1\mathbf{x},S^1\mathbf{x}') \models
\reach{\VASSYM{}^2}(S^2\mathbf{x},S^2\mathbf{x}')$.

Thus, our problem can be stated as follows: given a transition formula
$F$, synthesize a \VAS{} abstraction $(S,\VASSYM{})$ of $F$ such that
$(S,\VASSYM{})$ is \textit{best} in the sense that we have
$(S,\VASSYM{}) \preceq (\widetilde{S},\widetilde{\VASSYM{}})$ for any
\VAS{} abstraction $(\widetilde{S},\widetilde{\VASSYM{}})$ of $F$.
A solution to this problem is given in Algorithm~\ref{alg:core}.

\begin{algorithm}
\SetKwInOut{Input}{input}\SetKwInOut{Output}{output}
\Input{Transition formula $F$ of dimension $n$}
\Output{\VAS{} abstraction of $F$; Best \VAS{} abstraction if $F$ in $\LRA$}
Skolemize existentials of $F$\;
$(S, \VASSYM{}) \gets (I_n, \emptyset)$\tcp*{$(I_n,\emptyset)$ is least in $\preceq$ order}
$\Gamma \gets F$\;
\While{$\Gamma$ is satisfiable\label{alg:core:unsat}}{
  Let $M$ be a model of $\Gamma$\;
  $C \gets $ cube of the DNF of $F$ with $M \models C$\;
  $(S,\VASSYM{}) \gets (S,\VASSYM{}) \sqcup \hat{\alpha}(C)$\label{alg:core:combine}\;
  $\Gamma \gets \Gamma \wedge \lnot \gamma(S,\VASSYM{})$
}
\Return{$(S,\VASSYM{})$}
\caption{\texttt{abstract-VASR}(F) \label{alg:core}}
\end{algorithm}

Algorithm~\ref{alg:core} follows the familiar pattern of an
AllSat-style loop.  The algorithm takes as input a transition formula
$F$. It maintains a \VAS{} abstraction $(S,V)$ and a formula $\Gamma$,
whose models correspond to the transitions of $F$ that are
\textit{not} simulated by $(S,V)$.  The idea is to build $(S,V)$
iteratively by sampling transitions from $\Gamma$, augmenting $(S,V)$
to simulate the sample transition, and then updating $\Gamma$
accordingly.  We initialize $(S,V)$ to be $(I_n, \emptyset)$, the
canonical least \VAS{} abstraction in $\preceq$ order, and $\Gamma$ to
be $F$ (i.e., $(I_n,\emptyset)$ does not simulate any transitions of
$F$).  Each loop iteration proceeds as follows.  First, we sample a
model $M$ of $\Gamma$ (i.e., a transition that is allowed by $F$ but
not simulated by $(S,\VASSYM{})$).  We then generalize that transition
to a set of transitions by using $M$ to select a cube $C$ of the DNF
of $F$ that contains $M$.  Next, we use the procedure described in
Section~\ref{sec:conjunctive} to compute a \VAS{} abstraction
$\hat{\alpha}(C)$ that simulates the transitions of $C$.  We then
update the \VAS{} abstraction $(S,V)$ to be the least upper bound of
$(S,V)$ and $\hat{\alpha}(C)$ (w.r.t. $\preceq$ order) using the
procedure described in Section~\ref{sec:join} (line
\ref{alg:core:combine}).  Finally, we block any transition simulated
by the least upper bound (including every transition in $C$) from
being sampled again by conjoining $\lnot \gamma(S,\VASSYM)$ to
$\Gamma$.  The loop terminates when $\Gamma$ is unsatisfiable, in
which case we have that $F \Vdash_S \VASSYM{}$.
Theorem~\ref{thm:vas-abstraction} gives the correctness statement
for this algorithm.

\begin{restatable}{theorem}{vasabstraction}
  \label{thm:vas-abstraction}
  Given a transition formula $F$, Algorithm~\ref{alg:core} computes a
  simulation $S$ and \VAS{} $V$ such that $F \Vdash_S V$.  Moreover,
  if $F$ is in $\LRA{}$, Algorithm~\ref{alg:core} computes a
  \emph{best} \VAS{} abstraction of $F$.
\end{restatable}

The proof of this theorem as well as the proofs to all subsequent theorems, lemmas, and propositions
are in the extended version of this paper \cite{extended}.

\subsection{Abstracting conjunctive transition formulas}\label{sec:conjunctive}

This section shows how to compute a \VAS{} abstraction for a
consistent \textit{conjunctive} formula. When the input formula is in $\LRA$, the computed
\VAS{} abstraction will be a best \VAS{} abstraction of the input formula.  The intuition is that, since
$\LRA{}$ is a convex theory, a best \VAS{} abstraction consists of a
single transition.  For $\LIRA{}$ formulas, our procedure produces a
\VAS{} abstract that is not guaranteed to be best, precisely because
$\LIRA{}$ is not convex.

Let $C$ be consistent, conjunctive transition formula.  Observe that
the set ${\textit{Res}_C \defeq \{ \tuple{\vec{s}, a} : C \models \vec{s} \cdot
  \vec{x'} = a\}}$, which represents linear combinations of variables
that are \textit{reset} across $C$, forms a vector space.  Similarly, the set
$\textit{Inc}_C = \{
\tuple{\vec{s}, a} : C \models \vec{s} \cdot \vec{x'} = \vec{s} \cdot
\vec{x} + a\}$, which
represents linear combinations of variables that are
\textit{incremented} across $C$, forms a vector space.   We compute bases for both $\textit{Res}_C$ and $\textit{Inc}_C$, say $\{
\tuple{\vec{s}_1, a_1}, ..., \tuple{\vec{s}_m,a_m} \}$ and $\{
\tuple{\vec{s}_{m+1}, a_{m+1}}, ..., \tuple{\vec{s}_d,a_d} \}$,
respectively.  We define $\hat{\alpha}(C)$ to be the \VAS{}
abstraction $\hat{\alpha}(C) \defeq (S, \{(\res, \adv)\})$, where
\[ S \defeq \begin{bmatrix} \vec{s}_1\\\vdots\\\vec{s}_d\end{bmatrix} \hspace*{1cm} \res \defeq [\underbrace{0 \cdots 0}_{m\text{ times}} \overbrace{1 \cdots 1}^{(d-m)\text{ times}}\!\!\!]\ \hspace*{1cm} \adv \defeq \begin{bmatrix} a_1\\\vdots\\a_d\end{bmatrix}. \]


\begin{example}
  Let $C$ be the formula $x' = x + y \land y' = 2y \land w' = w \land w = w + 1 \land z' = w$.
   The vector space of resets has basis $\{\tuple{\begin{bmatrix}0 & 0 &
      -1 & 1\end{bmatrix}, 0}\}$ (representing that $z - w$ is reset to 0).
  The vector space of increments has basis
  $\{\tuple{\begin{bmatrix}1 & -1 & 0 & 0\end{bmatrix}, 0},
  \tuple{\begin{bmatrix}0 & 0 & 1 & 0\end{bmatrix}, 0}, 
  \tuple{\begin{bmatrix}0 & 0 & -1 & 1\end{bmatrix}, 1}\}$ (representing
  that the difference $x - y$ does not change, the difference $z - w$ increases by 1, 
  and the variable $w$ does not change).
  A best
  abstraction of $C$ is thus the four-dimensional \VAS{}
  \[ \VASSYM = \left\{
  \left(\begin{bmatrix} 0 \\ 1 \\ 1 \\ 1 \end{bmatrix}, \begin{bmatrix}0 \\ 0
    \\0 \\ 1\end{bmatrix}\right) \right\}, S = \begin{bmatrix}0
      & 0 & -1 & 1\\1 & -1 & 0 & 0\\0 & 0 & 1 & 0\\ 0 & 0 & -1 & 1\end{bmatrix}.
      \]
      In particular,
      notice that since the term $z - w$ is both incremented and
      reset, it is represented by two different dimensions in
      $\hat{\alpha}(C)$.
\end{example}

\begin{restatable}{proposition}{alphahat}
\label{prop:alphahat}
  For any consistent, conjunctive transition formula $C$,
  $\hat{\alpha}(C)$ is a \VAS{} abstraction of $C$.  If $C$ is
  expressed in $\LRA$, then $\hat{\alpha}(C)$ is best.
\end{restatable}

\subsection{Computing least upper bounds}\label{sec:join}
This section shows how to compute least upper bounds w.r.t. the
$\preceq$ order.

By definition of the $\preceq$ order, if $(S,\VASSYM{})$ is an
upper bound of $(S^1,\VASSYM{}^1)$ and $(S^2,\VASSYM{}^2)$, then there
must exist matrices $T^1$ and $T^2$ such that $T^1S^1 = S = T^2S^2$,
$V^1 \Vdash_{T^1} V$, and $V^2 \Vdash_{T^2} V$.  As we shall see, if
$(S,V)$ is a \textit{least} upper bound, then it is completely
determined by the matrices $T^1$ and $T^2$.  Thus, we shift our
attention to computing simulation matrices $T^1$ and $T^2$ that induce
a least upper bound.

In view of the desired equation $T^1S^1 = S = T^2S^2$, let us consider
the constraint $T^1S^1 = T^2S^2$ on two \textit{unknown} matrices
$T^1$ and $T^2$. Clearly, we have $T^1S^1 = T^2S^2$ iff each
$(T^1_i,T^2_i)$ belongs to the set $\mathcal{T} \defeq \{ (\vec{t}^1,
\vec{t}^2) : \vec{t}^1S^1 = \vec{t}^2S^2 \}$.  Observe that
$\mathcal{T}$ is a vector space, so there is a \textit{best} solution
to the constraint $T^1S^1 = T^2S^2$: choose $T^1$ and $T^2$ so that
the set of all row pairs $(T^1_i,T^2_i)$ forms a basis for
$\mathcal{T}$.  In the following, we use $\textit{pushout}(S^1,S^2)$
to denote a function that computes such a \textit{best} $(T^1, T^2)$.

While $\textit{pushout}$ gives a \textit{best} solution to the
equation $T^1S^1 = T^2S^2$, it is not sufficient for the purpose of
computing least upper bounds for \VAS{} abstractions, because $T^1$ and
$T^2$ may not respect the structure of the \VAS{} $V^1$ and $V^2$
(i.e., there may be no \VAS{} $V$ such that $V^1 \Vdash_{T^1} V$ and
$V^2 \Vdash_{T^2} V$).  Thus, we must further constrain our problem by
requiring that $T^1$ and $T^2$ are \textit{coherent} with respect to
$V^1$ and $V^2$ (respectively).

\begin{definition}\label{def:coherence}
  Let $\VASSYM$ be a $d$-dimensional \VAS{}.  We say that $i,j \in
  \{1,...,d\}$ are \textbf{coherent dimensions} of $\VASSYM$ if for
  all transitions $(\res,\adv) \in \VASSYM{}$ we have $r_i = r_j$
  (i.e., every transition of $\VASSYM{}$ that resets $i$ also resets
  $j$ and vice versa).  We denote that $i$ and $j$ are coherent
  dimensions of $\VASSYM{}$ by writing $i \equiv_{\VASSYM{}} j$, and
  observe that $\equiv_{\VASSYM{}}$ forms an equivalence relation on
  $\{1,...,d\}$.  We refer to the equivalence classes of
  $\equiv_{\VASSYM{}}$ as the \textbf{coherence classes} of $\VASSYM{}$.

  A matrix $T \in \mathbb{Q}^{e \times
    d}$ \textbf{is coherent with respect to} $\VASSYM{}$ if and only if each
  of its rows have non-zero values only in the dimensions corresponding to a single coherence class of $\VASSYM{}$.


  \end{definition}

For any $d$-dimensional \VAS{} $\VASSYM{}$ and coherence class $C =
\{c_1,...,c_k\}$ of $\VASSYM{}$, define $\proj{C}$ to be the $k \times
d$ dimensional matrix whose rows are $\vec{e}_{c_1}, ...,
\vec{e}_{c_k}$.  Intuitively, $\proj{C}$ is a projection onto the set
of dimensions in $C$.


Coherence is a necessary and sufficient condition for linear
simulations between \VAS{} in a sense described in
Lemmas~\ref{lem:coherence} and~\ref{lem:image}.

\begin{restatable}{lemma}{coherence}\label{lem:coherence}
  Let $\VASSYM{}^1$ and $\VASSYM{}^2$ be \VAS{} (of dimension $d$ and
  $e$, respectively), and let $T \in \mathbb{Q}^{e \times d}$ be a
  matrix such that $\VASSYM{}^1 \Vdash_T \VASSYM{}^2$.
  Then $T$ must be coherent with respect to $\VASSYM{}^1$.
\end{restatable}


Let $\VASSYM{}$ be a $d$-dimensional \VAS{} and let $T \in
\mathbb{Q}^{e \times d}$ be a matrix that is coherent with respect to
$\VASSYM{}$ and has no zero rows.  Then there is a (unique)
$e$-dimensional \VAS{} $\image{\VASSYM{}}{T}$ such that its transition
relation $\rightarrow_{\image{\VASSYM{}}{T}}$ is equal to
$\{(T\vec{u},T\vec{v}) : \vec{u} \rightarrow_{\VASSYM} \vec{v}\}$ (the
image of $\VASSYM{}$'s transition relation under $T$).  This \VAS{}
can be defined by:
\[\image{\VASSYM{}}{T} \defeq \{
  (\rim{T}{\res}, T\adv) : (\res,\adv) \in \VASSYM{} \} \] where
$\rim{T}{\res}$ is the reset vector $\res$ translated along $T$ (i.e.,
$(\rim{T}{\res})_i = r_j$ where $j$ is an arbitrary choice among
dimensions for which $T_{ij}$ is non-zero---at least one such $j$
exists because the row $T_i$ is non-zero by assumption, and the choice
of $j$ is arbitrary because all such $j$ belong to the same coherence
class by the assumption that $T$ is coherent with respect to
$\VASSYM{}$).


\begin{restatable}{lemma}{lemimage}
\label{lem:image}
Let $\VASSYM{}$ be a $d$-dimensional \VAS{} and let $T \in
\mathbb{Q}^{e \times d}$ be a matrix that is coherent with respect to $\VASSYM{}$
and has no zero rows.  Then the transition relation of
$\image{\VASSYM{}}{T}$ is the image
of $\VASSYM{}$'s transition relation under $T$ (i.e.,
$\rightarrow_{\image{\VASSYM{}}{T}}$ is equal to $\{(T\vec{u},T\vec{v})
  : \vec{u} \rightarrow_{\VASSYM} \vec{v}\}$).
\end{restatable}

Finally, prior to describing our least upper bound algorithm, we must define a
technical condition that is both assumed and preserved by the
procedure:
\begin{definition}
  A \VAS{} abstraction $(S,\VASSYM{})$ is \textbf{normal} if there is no non-zero vector
  $\vec{z}$ that is coherent with respect to $\VASSYM{}$ such
  that $\vec{z}S = 0$ (i.e., the rows of $S$ that correspond
  to any coherence class of $\VASSYM{}$ are linearly independent).
\end{definition}
Intuitively, a \VAS{} abstraction that is \textit{not} normal contains
information that is either inconsistent or redundant.

\begin{algorithm}[t]
  \SetKwInOut{Input}{input}\SetKwInOut{Output}{output}
  \Input{Normal \VAS{} abstractions $(S^1,\VASSYM{}^1)$ and $(S^2,\VASSYM{}^2)$ of equal concrete dimension}
  \Output{Least upper bound (w.r.t. $\preceq$) of $(S^1,\VASSYM{}^2)$ and $(S^1,\VASSYM{}^2)$}
  $S, T^1, T^2 \gets$ empty matrices\;
  \ForEach{\text{coherence class } $C^1 \text{ of } \VASSYM{}^1$}{
  \ForEach{\text{coherence class } $C^2 \text{ of } \VASSYM{}^2$}{
  $(U^1, U^2) \gets \textit{pushout}(\proj{C^1}{S^1},  \proj{C^2}{S^2})$\;
  $S \gets \begin{bmatrix} S\\ U^1 \proj{C^1}{S^1} \end{bmatrix}$;
$T^1 \gets \begin{bmatrix}
  T^1\\
  U^1 \proj{C^1}
\end{bmatrix}$; $T^2 \gets \begin{bmatrix}
  T^2\\
  U^2 \proj{C^2}
\end{bmatrix}$\;
}
}
$\VASSYM{} \gets \image{\VASSYM{}^1}{T^1} \cup \image{\VASSYM{}^2}{T^2}$\;
\Return{$(S, \VASSYM{})$}
\caption{$(S^1,\VASSYM{}^1) \sqcup (S^2,\VASSYM{}^2)$ \label{alg:join}}
\end{algorithm}




We now present a strategy for computing least upper bounds of \VAS{}
abstractions.  Fix (normal) \VAS{} abstractions $(S^1,\VASSYM{}^1)$
and $(S^2,\VASSYM^2{})$.  Lemmas~\ref{lem:coherence}
and~\ref{lem:image} together show that a pair of matrices
$\widetilde{T}^1$ and $\widetilde{T}^2$ induce an upper bound (not
necessarily \textit{least}) on $(S^1,\VASSYM{}^1)$ and
$(S^2,\VASSYM^2{})$ exactly when the following conditions hold: (1)
$\widetilde{T}^1S^1 = \widetilde{T}^2S^2$, (2) $\widetilde{T}^1$ is
coherent w.r.t. $V^1$, (3) $\widetilde{T}^2$ is coherent w.r.t. $V^2$,
and (4) neither $\widetilde{T}^1$ nor $\widetilde{T}^2$ contain zero
rows.  The upper bound induced by $\widetilde{T}^1$ and
$\widetilde{T}^2$ is given by
\[
\textit{ub}(\widetilde{T}^1, \widetilde{T}^2) \defeq (\widetilde{T}^1 S^1, \image{\VASSYM{}^1}{\widetilde{T}^1} \cup \image{\VASSYM{}^2}{T^2})\ .
\]
We now consider how to compute a \textit{best} such $\widetilde{T}^1$ and $\widetilde{T}^2$.
Observe that conditions (1),(2), and (3) hold exactly when for each
row $i$, $(\widetilde{T}^1_i, \widetilde{T}^2_i)$ belongs to the set
\[ \mathcal{T} \defeq \{ (\vec{t}^1, \vec{t}^2) :
\vec{t}^1S^1 = \vec{t}^2S^2 \land \vec{t}^1 \textit{coherent w.r.t.} V^1 \land \vec{t}^1 \textit{coherent w.r.t.} V^2 \}\ .
\]
Since a row vector $\vec{t}^i$ is coherent w.r.t $V^i$ iff its
non-zero positions belong to the same coherence class of $V^i$
(equivalently, $\vec{t}^i = \bar{\vec{t}}^i \proj{C^i}$ for some coherence
class $C^i$ and vector $\bar{\vec{t}}^i$), we have $\mathcal{T} =
\bigcup_{C^1,C^2} \mathcal{T}(C^1,C^2)$, where
the union is over all coherence classes $C^1$ of $V^1$ and
$C^2$ of $V^2$, and 
\[ \mathcal{T}(C^1,C^2) \defeq \{ (\bar{\vec{t}}^{i,1} \proj{C^1}, \bar{\vec{t}}^{i,1}\proj{C^2}) : \bar{\vec{t}}^{i,1} \proj{C^1}S^1 = \bar{\vec{t}}^{i,1}\proj{C^2}S^2 \}\ .
\]
Observe that each $\mathcal{T}(C^1,C^2)$ is a vector space, so we
can compute a pair of matrices $T^1$ and $T^2$ such that the rows
$(T^1_i, T^2_i)$ collectively form a basis for each
$\mathcal{T}(C^1,C^2)$.  Since $(S^1,\VASSYM{}^1)$ and
$(S^2,\VASSYM^2{})$ are normal (by assumption), neither $T^1$ nor
$T^2$ may contain zero rows (condition (4) is satisfied).  Finally, we
have that $\textit{ub}(T^1,T^2)$ is the \textit{least} upper bound of
$(S^1,\VASSYM{}^1)$ and $(S^2,\VASSYM^2{})$.  Algorithm~\ref{alg:join}
is a straightforward realization of this strategy.

\begin{restatable}{proposition}{join}
\label{prop:join}
  Let $(S^1,\VASSYM{}^1)$ and $(S^2,\VASSYM{}^2)$ be normal 
  \VAS{}
  abstractions of equal concrete dimension.  Then the \VAS{} abstraction
  $(S,\VASSYM{})$ computed by Algorithm~\ref{alg:join} is normal
  and is a least upper bound of $(S^1,\VASSYM{}^2)$ and
  $(S^2,\VASSYM{}^2)$.
\end{restatable}


\section{Control Flow and \VASS{}} \label{sec:vass}

In this section, we give a method for improving the precision of our
loop summarization technique by using \VASS{}; that is, \VAS{}
extended with control states.  While \VAS{}s over-approximate control
flow using non-determinism, \VASS{}s allow us to analyze phenomena
such as oscillating and multi-phase loops.

We begin with an example that demonstrates the precision gained by
\VASS{}.  The loop in Figure \ref{fig:osc-loop} oscillates between (1)
incrementing variable $i$ by $1$ and (2) incrementing both variables
$i$ and $x$ by $1$.  Suppose that we wish to prove that, starting with
the configuration $x = 0 \land i = 1$, the loop maintains the
invariant that $2x \leq i$.  The (best) \VAS{} abstraction of the
loop, pictured in Figure~\ref{fig:osc-vasr}, over-approximates the
control flow of the loop by treating the conditional branch in the
loop as a non-deterministic branch. This over-approximation may
violate the invariant $2x \leq i$ by repeatedly executing the path
where both variables are incremented.  On the other hand, the \VASS{}
abstraction of the loop pictured in Figure~\ref{fig:cfg} captures the
understanding that the loop must oscillate between the two paths.  The
loop summary obtained from the reachability relation of this \VASS{}
is powerful
enough to prove the invariant $2x \leq i$ holds (under the precondition
$x = 0 \land i = 1$).

\begin{figure}
  \centering
\begin{subfigure}[b]{3cm}
  \textbf{int} \texttt{x = 0; i = 1}\\
  \textbf{while} \texttt{(*)} \textbf{do}\\
  \hspace*{10pt}\textbf{if} \texttt{i\%2 == 0} \textbf{then}\\
  \hspace*{20pt}\texttt{i := i + 1}\\
  \hspace*{10pt}\textbf{else}\\
  \hspace*{20pt}\texttt{i := i + 1}\\
  \hspace*{20pt}\texttt{x := x + 1}\\
  \subcaption{\label{fig:osc-loop}Oscillating loop}
\end{subfigure}
\begin{subfigure}[b]{4cm}
    \centering
    $\left\{ \begin{array}{l}\begin{bmatrix}i\\x\end{bmatrix} \mapsto \begin{bmatrix}i+1\\x+1\end{bmatrix},\\[10pt] \begin{bmatrix}i\\x\end{bmatrix} \mapsto \begin{bmatrix}i+1\\x\end{bmatrix}\end{array}\right\}$
  \subcaption{\label{fig:osc-vasr}\VAS{} abstraction.}
\end{subfigure}
\begin{subfigure}[b]{5cm}
\centering
\begin{tikzpicture}[>=stealth]
\node[ellipse,draw](1){\texttt{i\%2 == 0}};
\node[ellipse,draw,right of=1,node distance=3cm](2){\texttt{i\%2 == 1}};
\draw[->](1) to [out=90,in=90]node[fill=white,yshift=-5pt]{$\begin{bmatrix}i\\x\end{bmatrix}\mapsto \begin{bmatrix}i+1\\x\end{bmatrix}$}(2);
\draw[->](2) to [out=270,in=270]node[fill=white,yshift=5pt]{$\begin{bmatrix}i\\x\end{bmatrix}\mapsto \begin{bmatrix}i+1\\x+1\end{bmatrix}$}(1);
\end{tikzpicture}
  \subcaption{\label{fig:cfg}\VASS{} abstraction.}
\end{subfigure}
\caption{\label{fig:vassexamp}An oscillating loop and its representation as a \VAS{} and \VASS{}.}
\end{figure}

 

\subsection{Technical details}

In the following, we give a method for over-approximating the
transitive closure of a transition formula $F(\vec{x},\vec{x}')$ using
a \VASS{}.  We start by defining \textit{predicate \VASS{}}, a
variation of \VASS{} with control states that correspond to disjoint
state predicates (where the states intuitively belong to the
transition formula $F$ rather than the \VASS{} itself).  We extend
linear simulations and best abstractions to predicate \VASS{}, and give
an algorithm for synthesizing best predicate \VASS{} abstractions (for
a given set of predicates).  Finally, we give an end-to-end algorithm
for over-approximating the transitive closure of a transition formula.

\begin{definition}  \label{def:npred}
  A \textbf{predicate \VASS{}} over $\vec{x}$ is a \VASS{} $\SYMVASS{}
  = (P, E)$, such that each control state is a predicate over the
  variables $\vec{x}$ and the predicates in $P$ are pairwise
  inconsistent (for all $p \neq q \in P$, $p \land q$ is
  unsatisfiable).
\end{definition}

We extend linear simulations to predicate \VASS{} as follows:
\begin{itemize}
\item Let $F(\vec{x},\vec{x}')$ be an $n$-dimensional transition formula and let $\SYMVASS{} = (P, E)$ be an $m$-dimensional
  \VASS{} over $\vec{x}$.  We say that a linear
  transformation $S : \mathbb{Q}^{m \times n}$
  is a
  linear simulation from $F$ to $\SYMVASS{}$ if for all $\vec{u},\vec{v}
  \in \mathbb{Q}^n$ such that $\vec{u} \rightarrow_F \vec{v}$,
 (1) there is a (unique) $p \in P$ such that $p(\vec{u})$ is valid
  (2) there is a (unique) $q \in P$ such that $q(\vec{v})$ is valid, and (3)
  $(p, S\vec{u}) \rightarrow_{\SYMVASS{}} (q, S\vec{v})$.
\item Let $\SYMVASS{}^1 = (P^1, E^1)$ and $\SYMVASS{}^2 = (P^2, E^2)$ be predicate \VASS{}s over $\vec{x}$ (for some $\vec{x}$) of dimensions $d$ and $e$, respectively.
  We say that a linear transformation $S : \mathbb{Q}^{e \times d}$
   is a linear simulation from $\SYMVASS{}^1$ to
  $\SYMVASS{}^2$ if for all $p^1,q^1 \in P^1$ and for all $\vec{u},\vec{v} \in
  \mathbb{Q}^d$ such that $(p^1, \vec{u}) \rightarrow_{\SYMVASS{}^1} (q^1,
  \vec{v})$, there exists (unique) $p^2, q^2 \in P^2$
  such that (1) $(p^2, S\vec{u}) \rightarrow_{\SYMVASS{}^2}
  (q^2,S\vec{v})$, (2) $p^1 \models p^2$, and (3) $q^1 \models q^2$.
\end{itemize}


We define a \VASS{} abstraction over $\vec{x} = x_1,...,x_n$ to be a pair $(S,\SYMVASS{})$ consisting of
a rational matrix $S \in \mathbb{Q}^{d \times n}$ and a
predicate \VASS{} of dimension $d$ over $\vec{x}$.  We extend the
simulation preorder $\preceq$ to \VASS{} abstractions in the natural
way.  Extending the definition of ``best'' abstractions requires more
care, since we can always find a ``better'' \VASS{} abstraction
(strictly smaller in $\preceq$ order) by using a finer set of predicates.  However, if we consider only predicate \VASS{} that share
the same set of control states, then best abstractions do exist and can be computed using
Algorithm~\ref{alg:vasscore}.

\begin{algorithm}
\SetKwInOut{Input}{input}\SetKwInOut{Output}{output}
\Input{Transition formula $F(\vec{x},\vec{x}')$, set of pairwise-disjoint predicates $P$ over $\vec{x}$ such that for all $\vec{u},\vec{v}$ with $\vec{u} \rightarrow_F \vec{v}$, there exists $p,q \in P$ with $p(\vec{u})$ and $q(\vec{v})$ both valid}
\Output{Best \VASS{} abstraction of $F$ with control states $P$}
For all $p, q \in P$, let $(S_{p,q},\VASSYM_{p,q}) \gets \texttt{abstract-VASR}(p(\vec{x})  \land F(\vec{x},\vec{x}') \land q(\vec{x}'))$\;
$(S,V) \gets$  least upper bound of all $(S_{p,q},\VASSYM_{p,q})$\;
For all $p, q \in P$, let $T_{p,q} \gets$ the simulation matrix from $(S_{p,q},\VASSYM_{p,q})$ to $(S,V)$\;
$E = \{ (p, \res, \adv, q) : p, q \in P, (\res,\adv) \in \image{\VASSYM_{p,q}}{T_{p,q}} \}$\;
\Return{$(S, (P, E))$}
\caption{\texttt{abstract-VASRS}$(F,P)$} \label{alg:vasscore}
\end{algorithm}

Algorithm~\ref{alg:vasscore} works as follows: first, for each pair of
formulas $p, q \in P$, compute a best \VAS{} abstraction of the
formula $p(\vec{x}) \wedge F(\vec{x},\vec{x}') \wedge q(\vec{x}')$ and
call it $(S_{p,q},\VASSYM{}_{p,q})$. $(S_{p,q},\VASSYM{}_{p,q})$
over-approximates the transitions of $F$ that begin in a program state
satisfying $p$ and end in a program state satisfying $q$.  Second, we
compute the least upper bound of all \VAS{} abstractions
$(S_{p,q},\VASSYM{}_{p,q})$ to get a \VAS{} abstraction $(S,V)$ for
$F$.  As a side-effect of the least upper bound computation, we obtain
a linear simulation $T_{p,q}$ from $(S_{p,q},\VASSYM_{p,q})$ to
$(S,V)$ for each $p,q$. A best \VASS{} abstraction of $F(\vec{x},\vec{x}')$ with
control states $P$ has $S$ as its
simulation matrix and has the image of $V_{p,q}$ under $T_{p,q}$
as the edges from $p$ to $q$.



\begin{restatable}{proposition}{vass}
\label{prop:vass}
  Let $F(\vec{x},\vec{x}')$ be a transition formula and let $P$ be a set
  of pairwise inconsistent control states over $\vec{x}$ such that for each 
  transition $\vec{u} \rightarrow_F \vec{v}$, there exists a control states $p,q \in P$
  such that $\vec{u} \models p$ and $\vec{v} \models q$.
  Algorithm~\ref{alg:vasscore} computes a predicate \VASS{} abstraction of $F$ with control states $P$.
  Moreover, if $F$ is in $\LRA$, algorithm~\ref{alg:vasscore} computes a best predicate \VASS{} abstraction
  of $F$ with control states $P$.
\end{restatable}

We now describe \texttt{iter-VASRS} (Algorithm~\ref{alg:clcomp}),
which uses \VASS{} to over-approximate the transitive closure of
transition formulas.  Towards our goal of \textit{predictable} program
analysis, we desire the analysis to be \textit{monotone} in the sense
that if $F$ and $G$ are transition formulas such that $F$ entails $G$,
then $\texttt{iter-VASRS}(F)$ entails $\texttt{iter-VASRS}(G)$.  A
sufficient condition to guarantee monotonicity of the overall analysis
is to require that the set of control states that we compute for $F$ is at least as fine as
the set of control states we compute for $G$.  We can achieve this by making the set of
control states $P$ of input transition formula $F(\vec{x},\vec{x}')$ equal to the set of
connected regions of the topological closure of $\exists \vec{x}'. F$
(lines~\ref{ln:start}-\ref{ln:end}).  Note that this set of predicates
may fail the contract of $\texttt{abstract-VASRS}$: there may exist a
transition $\vec{u} \rightarrow_F \vec{v}$ such that $\vec{v}
\not\models \bigvee P$ (this occurs when there is a state of $F$ with
no outgoing transitions). As a result, $(S,\mathcal{V}) =
\texttt{abstract-VASRS}(F,P)$ does not necessarily approximate $F$;
however, it \textit{does} over-approximate $F \land \bigvee P(\vec{x}')$. An
over-approximation of the transitive closure of $F$ can easily be
obtained from $\reach{\mathcal{V}}(S\vec{x},S\vec{x}')$ (the
over-approximation of the transitive closure of $F \land \bigvee P(\vec{x}')$
obtained from the \VASS{} abstraction ($S,\mathcal{V}$)) by
sequentially composing with the disjunction of $F$ and the identity relation
(line~\ref{ln:fixup}).

\begin{algorithm}
\SetKwInOut{Input}{input}\SetKwInOut{Output}{output}
\Input{Transition formula $F(\vec{x},\vec{x}')$}
\Output{Over-approximation of the transitive closure of $F$}
$P \gets$ topological closure of DNF of $\exists \vec{x}'.F$ (see \cite{LPAR:Monniaux2008})\; \label{ln:start}
\tcc{Compute connected regions}
\While{$\exists p_1,p_2 \in P$  with $p_1 \land p_2$ satisfiable}{
  $P \gets (P \setminus \{p_1,p_2\}) \cup \{ p_1 \lor p_2 \}$
} \label{ln:end}
$(S,\mathcal{V}) \gets \texttt{abstract-VASRS}(F, P)$\;
\Return{$\reach{\mathcal{V}}(S\vec{x},S\vec{x}') \circ (\vec{x}' = \vec{x} \lor F)$} \label{ln:fixup}
\caption{\texttt{iter-VASRS}$(F)$} \label{alg:clcomp}
\end{algorithm} 


\paragraph{Precision improvement}

The \texttt{abstract-VASRS} algorithm uses predicates to infer the
control structure of a \VASS{}, but after computing the \VASS{}
abstraction, \texttt{iter-VASRS} makes no further use of the
predicates (i.e., the predicates are irrelevant in the computation of
$\reach{\mathcal{V}}$).  Predicates can be used to improve
\texttt{iter-VASRS} as follows:  the reachability relation of a
\VASS{} is expressed by a formula that uses auxiliary variables to
represent the state at which the computation begins and ends
\cite{RP:HH2014}.  These variables can be used to encode that the
pre-state of the transitive closure must satisfy the predicate
corresponding to the begin state and the post-state must satisfy the
predicate corresponding to the end state.  As an example, consider the
Figure~\ref{fig:vassexamp} and suppose that we wish to prove the
invariant $x \leq 2i$ under the pre-condition $i = 0 \land x = 0$.
While this invariant holds, we cannot prove it because there is
counter example if the computation begins at $i\%2 == 1$.  By applying
the above improvement, we can prove that the computation must begin at
$i\%2 == 0$, and the invariant is verified.

\section{Evaluation} \label{sec:evaluation}

The goals of our evaluation is the answer the following questions:
\begin{itemize}
\item Are \VAS{} sufficiently expressive to be able to generate
  accurate loop summaries?
\item Does the \VASS{} technique improve upon the precision of \VAS{}?
\item Are the \VAS{}/\VASS{} loop summarization algorithms performant?
\end{itemize}

\begin{figure}[b]
\begin{center}
\begin{tabular}{|lc||cr|cr|cr|cr|cr|}
\hline
 & & \multicolumn{2}{c|}{\VAS{}} & \multicolumn{2}{c|}{\VASS{}} & \multicolumn{2}{c|}{CRA} & \multicolumn{2}{c|}{SeaHorn} & \multicolumn{2}{c|}{UltAuto}\\
 & & \#safe & time & \#safe & time & \#safe & time & \#safe & time & \#safe & time\\\hline\hline
C4B    & 35 & 21 & 37.9    & \textbf{31} & 35.4    & 27 & \textbf{33.1}    & 23 & 2434.4  & 25 & 3881.6  \\
HOLA   & 46 & 32 & 57.2    & 39 & 73.0    & \textbf{40} & \textbf{56.0}    & 35 & 2115.0  & 36 & 2995.9  \\
svcomp19-int & 84 & 68 & \textbf{86.9}    & \textbf{78} & 184.5   & 76 & 91.9    & 62 & 3038.0  & 64 & 6923.5  \\
\hline
\end{tabular}
\end{center}
\caption{Experimental results. \label{fig:experiments}}
\end{figure}

We implemented our loop summarization procedure and the compositional
whole-program summarization technique described in
Section~\ref{subsec:outline}.  We ran on a suite of 165 benchmarks,
drawn from the C4B \cite{PLDI:CHS15} and HOLA \cite{OOPSLA:DDLM13}
suites, as well as the safe, integer-only benchmarks in the loops
category of SV-Comp 2019 \cite{SVCOMP19}.  We ran each benchmark with
a time-out of 5 minutes, and recorded how many benchmarks were proved
safe by our \VAS{}-based technique and our \VASS{}-based technique.
For context, we also compare with CRA \cite{PACMPL:KCBR18} (a related
loop summarization technique), as well as SeaHorn \cite{CAV:GKKN15}
and UltimateAutomizer \cite{UltAuto} (state-of-the-art software model
checkers).  The results are shown in Figure~\ref{fig:experiments}.

The number of assertions proved correct using \VAS{} is comparable to
both SeaHorn and UltimateAutomizer, demonstrating that \VAS{} can
indeed model interesting loop phenomena.  \VASS{}-based summarization
significantly improves precision, proving the correctness of 93\% of
assertions in the svcomp suite, and more than any other tool in total. Note that
the most precise tool for each suite is not strictly better than
each of the other tools; in particular, there is only a single program
in the HOLA suite that neither \VASS{} nor CRA can prove safe.

CRA-based summarization is the most performant of all the compared
techniques, followed by \VAS{} and \VASS{}.  SeaHorn and
UltimateAutomizer employ abstraction-refinement loops, and so take
significantly longer to run the test suite.


\section{Related work} \label{sec:related-work}
\paragraph{Compositional analysis}
Our analysis follows the same high-level structure as compositional
recurrence analysis (CRA) \cite{FMCAD:FK15,PACMPL:KCBR18}.  Our
analysis differs from CRA in the way that it summarizes loops: we
compute loop summaries by over-approximating loops with vector addition
systems and computing reachability relations, whereas CRA computes loop
summaries by extracting recurrence relations and computing closed
forms.  The advantage of our approach is that is that we can use \VAS{} to
accurately model multi-path loops and can make theoretical guarantees
about the precision of our analysis; the advantage of CRA is its
ability to generate non-linear invariants.

\paragraph{Vector addition systems}
Our invariant generation method draws upon Haase and Halfon's
polytime procedure for computing the reachability relation of integer
vector addition systems with states and resets \cite{RP:HH2014}.
Generalization from the integer case to the rational case is
straightforward.  Continuous Petri nets \cite{PN:DA1987} are a related
generalization of vector addition systems, where time is taken to be
continuous (\VAS{}, in contrast, have rational state spaces but
discrete time).  Reachability for continuous Petri nets is computable polytime
\cite{FI:FH2015} and definable in $\LRA{}$ \cite{TACAS:BFHH2016}.

Sinn et al. present a technique for resource bound analysis
that is based on modeling programs by lossy vector addition system
with states \cite{CAV:SZV2014}.
Sinn et al. model
programs using vector addition systems with states over the natural
numbers, which enables them to use termination bounds for VASS to
compute upper bounds on resource usage.  In contrast, we use VASS with resets over
the rationals, which (in contrast to VASS over $\mathbb{N}$) have a $\LIRA{}$-definable
reachability relation, enabling us to summarize loops.  Moreover, Sinn
et al.'s method for extracting VASS models of programs is heuristic,
whereas our method gives precision guarantees.

\paragraph{Affine and polynomial programs}
The problem of \textit{polynomial} invariant generation has been
investigated for various program models that generalize \VAS{},
including solvable polynomial loops \cite{ISAAC:RCK2004}, (extended)
P-solvable loops \cite{TACAS:Kovacs2008,VMCAI:HJK2018}, and affine
programs \cite{LICS:HOPW2018}.  Like ours, these techniques are
\textit{predictable} in the sense that they can make theoretical
guarantees about invariant quality.  The kinds invariants that can be
produced using these techniques (conjunctions of polynomial equations)
is incomparable with those generated by the method presented in this paper ($\LIRA{}$
formulas).

\paragraph{Symbolic abstraction}  The main contribution of this paper
is a technique for synthesizing the best abstraction of a transition
formula expressible in the language of \VAS{} (with or without
states).  This is closely related to the \textit{symbolic abstraction}
problem, which computes the best abstraction of a formula within an
abstract domain.  The problem of computing best abstractions has been
undertaken for finite-height abstract domains \cite{VMCAI:RSY2004},
template constraint matrices (including intervals and octagons)
\cite{POPL:LAKGC2014}, and polyhedra \cite{CAV:TR2012,FMCAD:FK15}.  Our
best abstraction result differs in that (1) it is for a disjunctive
domain and (2) the notion of ``best'' is based on simulation rather
than the typical order-theoretic framework.



\bibliographystyle{abbrv}
\bibliography{references}

 \section{Proofs}\label{sec:app}

\alphahat*
\begin{proof}
 Let $C$ be a consistent, conjunctive transition formula and let
 $(S, \VASSYM{}) = \hat{\alpha}(C)$ be a \VAS{} abstraction.  Clearly we
 have that $C \Vdash_{S} \VASSYM{}$---it remains to show that
 $\hat{\alpha}(C)$ is \textit{best}.  Suppose that $(\widetilde{S}, \widetilde{\VASSYM{}})$
 is a \VAS{} abstraction such that $C
 \Vdash_{\widetilde{S}} \widetilde{\VASSYM{}}$.  We must show that there exists a
 linear simulation $T$ such that $\VASSYM{} \Vdash_T \widetilde{\VASSYM{}}$
 and $\widetilde{S} = T S$.

 First, we show that there is a single transition $(\widetilde{\res},\widetilde{\adv})
 \in \widetilde{\VASSYM{}}$ that simulates $C$ (i.e., $C \Vdash_{\widetilde{S}}
 \{(\widetilde{\res},\widetilde{\adv})\}$).  This follows essentially from the fact
 that linear rational arithmetic is a convex theory; for completeness,
 we make an explicit argument. By the
 well-ordering principle, it is sufficient to prove that if $\widetilde{\VASSYM{}} = \{
 (\widetilde{\res_1},\widetilde{\adv_1}),...,(\widetilde{\res_n},\widetilde{\adv_n}) \}$ is a \VAS{} such
 that $C \Vdash_{\widetilde{S}} \widetilde{\VASSYM{}}$ and if there is no proper subset $U$ of $\widetilde{\VASSYM{}}$
 such that $C \not\Vdash_{\widetilde{S}} U$, then we must have $n=1$.  For a
 contradiction, suppose $n>1$, and let $U_1
 =\{(\widetilde{\res_1},\widetilde{\adv_1})\}$ and $U_2 = \{
 (\widetilde{\res_2},\widetilde{\adv_2}),...,(\widetilde{\res_{n}},\widetilde{\adv_{n}}) \}$.  Since
 $C \not\Vdash_{\widetilde{S}} U_1$, there is a
 transition $\vec{u}_1 \rightarrow_C \vec{v}_1$ such that
 $\widetilde{S}\vec{u}_1 \not\rightarrow_{U_1}
 \widetilde{S}\vec{v}_1$.  Since $C \not\Vdash_{\widetilde{S}} U_2$ there is a
 transition $\vec{u}_2 \rightarrow_C \vec{v}_2$ such that
 $\widetilde{S}\vec{u}_2 \not\rightarrow_{U_2} \widetilde{S}\vec{v}_2$.
 Geometrically, $C$ forms a convex polyhedron to which the points
 $(\vec{u}_1,\vec{v}_1)$ and $(\vec{u}_2,\vec{v}_2)$ belong.  By
 convexity, every point on the line segment from
 $(\vec{u}_1,\vec{v}_1)$ and $(\vec{u}_2,\vec{v}_2)$ belongs to $C$;
 that is, for all $k \in [0,1]$ we have $(k\vec{u}_1 + (1-k)\vec{u}_2)
 \rightarrow_C (k\vec{v}_1 + (1-k)\vec{v}_2)$.  Since there are
 infinitely many transitions along the line segment and each one must
 have a corresponding transition in $\widetilde{\VASSYM{}}$ that simulates it, there must
 exist some $i \in \{1,...,n\}$ such that the set $A_i$ of transitions
 that are simulated by transition $(\widetilde{\res_i},\widetilde{\adv_i})$,
 \[ A_i = \{ (\vec{u},\vec{v}) : \widetilde{S}\vec{u} \rightarrow_{(\widetilde{\res_i},\widetilde{\adv_i})} \widetilde{S}\vec{v} \} = \{ (\vec{u},\vec{v}) : \widetilde{S}\vec{v} = \widetilde{\res_i} 
 \hadamard{} \widetilde{S}\vec{u} + \widetilde{\adv_i} \}\ , \]
 contains at least two points on the line segment.  Since $A_i$ is an
 affine space and contains at least two points on the line segment, it
 must contain all points on the entire line that connects
 $(\vec{u}_1,\vec{v}_1)$ and $(\vec{u}_2,\vec{v}_2)$ (and in
 particular the points $(\vec{u}_1,\vec{v}_1)$ and
 $(\vec{u}_2,\vec{v}_2)$ themselves).  Since $\widetilde{S}\vec{u}_1
 \rightarrow_{(\widetilde{\res_i},\widetilde{\adv_i)}} \widetilde{S}\vec{v}_1$ and (by
 construction) $\widetilde{S}\vec{u}_1 \not\rightarrow_{U_1}
\widetilde{S}\vec{v}_1$, we cannot have $i=1$.  Since $\widetilde{S}\vec{u}_2
 \rightarrow_{(\widetilde{\res_i},\widetilde{\adv_i})} \widetilde{S}\vec{v}_2$ and (by
 construction) $\widetilde{S}\vec{u}_2 \not\rightarrow_{U_2}
\widetilde{S}\vec{v}_2$ we also cannot have $i \neq 1$, a contradiction.

Next we construct a matrix $T$ such that
$T S = \widetilde{S}$ and that $\VASSYM{} \Vdash_{T} \widetilde{\VASSYM{}}$.
Recall that $\hat{\alpha}(C)$ is defined to be
$(S,\{(\res,\adv)\})$, with
\[ S \defeq \begin{bmatrix} \vec{s}_1\\\vdots\\\vec{s}_d\end{bmatrix} \hspace*{1cm} \res \defeq [\underbrace{0 \cdots 0}_{m\text{ times}} \overbrace{1 \cdots 1}^{(d-m)\text{ times}}\!\!\!]\ \hspace*{1cm} \adv \defeq \begin{bmatrix} a_1\\\vdots\\a_d\end{bmatrix} \]
    and where $\{\tuple{\vec{s}_1, a_1}, ..., \tuple{\vec{s}_m,a_m} \}$
    is a basis for the vector space $\textit{Res}_C \defeq \{ \tuple{\vec{s}, a} : C \models \vec{s} \cdot \vec{x'} = a\}$ 
    and $\{\tuple{\vec{s}_{m+1}, a_{m+1}}, ..., \tuple{\vec{s}_d,a_d} \}$ 
    is a basis for the vector space
    $\textit{Inc}_C = \{\tuple{\vec{s}, a} : C \models \vec{s} \cdot \vec{x'} = \vec{s} \cdot \vec{x} + a\}$.
    We form the $i$th row of the matrix $T$, $T_i$, as follows.
    Suppose that $\widetilde{\res_i} = 0$ (the case for $\widetilde{\res_i}=1$ is similar).  Since
    $C \Vdash_{\widetilde{S}} \{(\widetilde{\res},\widetilde{\adv})\}$, we have (using
    $\widetilde{S}_j$ to denote the $j$th row of $\widetilde{S}$)
    \begin{align*}
      C &\models \widetilde{S}\vec{x}' = \widetilde{\res} \hadamard{} \widetilde{S}\vec{x} + \widetilde{\adv}\\
      &\equiv \bigwedge_{j=1}^{d'} \widetilde{S}_j \cdot \vec{x}' = \widetilde{\res}_j\widetilde{S}_j' \cdot \vec{x} + \widetilde{\adv}_j\\
      &\models  \widetilde{S}_i \cdot \vec{x}'  = \widetilde{\res}_i\widetilde{S} \cdot \vec{x} + \widetilde{\adv}_i\\
      &=\widetilde{S}_i \cdot \vec{x}' = \widetilde{\adv}_i \,
    \end{align*}
    and thus we may conclude that $\tuple{\widetilde{S}_i,\widetilde{\adv_i}} \in \textit{Res}_C$. 
     It follows that there exist unique $t_1,...,t_m \in
    \mathbb{Q}$ such that $t_{1} \tuple{S_1, \adv_1} + \dotsi +
    t_{m} \tuple{S_m, \adv_m} = \tuple{\widetilde{S}_i, \widetilde{\adv}_i}$.  We take
    $T_i = \begin{bmatrix}t_1 & ... & t_m & 0 & ... & 0\end{bmatrix}$,
      and observe that $T_i S = \widetilde{S}_i$. Since this holds for all $i$, we
      have $T S=
      \widetilde{S}$.  For
      $\VASSYM{} \Vdash_{T} \widetilde{\VASSYM{}}$, we suppose that $\vec{u}
      \rightarrow_{\VASSYM{}} \vec{v}$ and prove that $T\vec{u} \rightarrow_{\widetilde{\VASSYM{}}} T\vec{v}$. 
      First, note that $\rim{T_i}\res = 0 = \widetilde{\res}_i$. Next observe that $T_i\adv =
      \widetilde{\adv}_i$. For each $i$, 
      $T_i\vec{v} = T_i(\res \hadamard{} \vec{u} + \adv)
        = T_i (\res\hadamard{}\vec{u}) + T_i \adv = (\rim{T_i}{\res}) 
        (T_i\vec{u}) + T_i \adv$, and therefore $T_i\vec{v} = \widetilde{\res}_i(T_i\vec{u}) + \widetilde{\adv}_i$. It follows that
      $T\vec{v} = \widetilde{\res} \hadamard{} T\vec{u} + \widetilde{\adv}$,
      and since $(\widetilde{\res},\widetilde{\adv}) \in \widetilde{\VASSYM{}}$ we have
      $T\vec{u} \rightarrow_{\widetilde{\VASSYM{}}} T\vec{v}$.
  \end{proof}

\coherence*
\begin{proof}
  Let $\VASSYM{}^1$ and $\VASSYM{}^2$ be \VAS{} (of dimension $d$ and $e$, respectively)
  and let $T \in \mathbb{Q}^{e \times d}$.
  Assume that $T$ is not coherent with respect to $\VASSYM{}^1$.
  Then there exist some $i,j,k$ such that $T_{ij}$ and $T_{ik}$ are
  non-zero and $j \not\equiv_{\VASSYM{}^1} k$.  The matrix defined by the 
  $i$th row of $T$
  is incoherent with respect to $\VASSYM{}^1$ and forms a linear
  simulation from $\VASSYM{}^1$ to the projection of $\VASSYM{}^2$
  onto its $i$th coordinate.  Thus, without loss of generality, we may
  assume that $e=1$ and $i=1$.

  Since $j \not\equiv_{\VASSYM{}^1} k$ there is some
  $(\res,\adv) \in \VASSYM{}^1$ such that $\resnv_j \neq \resnv_k$.
  Without loss of generality, assume $\resnv_j = 1$ and $\resnv_k = 0$.  We will
  show that there must be a transition $0 \rightarrow_{\VASSYM{}^2}
  z$ for all $z \in \mathbb{Q}$; this is a contradiction because
  $\rightarrow_{\VASSYM{}^2}$ is the transition relation of a \VAS{} and
  therefore finitely branching.  Let $z \in \mathbb{Q}$ be arbitrary.
  Let $\vec{e}_j$ and $\vec{e}_k$ denote the unit vectors in
  directions $j$ and $k$, respectively.  Since $T\mathbf{e}_j =
  T_{1j}$ and $T\mathbf{e}_k = T_{1k}$; both are non-zero by
  assumption.  Let $\mathbf{u} = \frac{z -
    T\adv}{T\mathbf{e}_j}\mathbf{e}_j +
  \frac{T\adv-z}{T\mathbf{e}_k}\mathbf{e}_k$, and let
  $\mathbf{v} = \res \hadamard{} \mathbf{u} + \adv$.  Since
  $\mathbf{u} \rightarrow_{\VASSYM{}^1} \mathbf{v}$ and $\VASSYM{}^1
  \Vdash_T \VASSYM{}^2$, we must have $T\mathbf{u}
  \rightarrow_{{\VASSYM{}^2}} T\mathbf{v}$.  Finally, calculate:
  \begin{center}
    \begin{minipage}{0.42\textwidth}
  \begin{align*}
    T\vec{u} &= T\left(\frac{z - T\adv}{T\mathbf{e}_j}\mathbf{e}_j +
    \frac{T\adv-z}{T\mathbf{e}_k}\mathbf{e}_k\right)\\
    &=\frac{z - T\adv}{T\mathbf{e}_j}T\mathbf{e}_j +
    \frac{T\adv-z}{T\mathbf{e}_k}T\mathbf{e}_k\\
    &= 0
  \end{align*}
    \end{minipage}
    \hfill
  \begin{minipage}{0.57\textwidth}
    \begin{align*}
      T\vec{v} &= T(\res \hadamard{} \mathbf{u} + \adv)\\
      &= T\left(\res \hadamard{} \left(\frac{z - T\adv}{T\mathbf{e}_j}\mathbf{e}_j +
      \frac{T\adv-z}{T\mathbf{e}_k}\mathbf{e}_k\right)+\adv\right)\\
      &= T\left(\frac{z - T\adv}{T\mathbf{e}_j}\resnv_j\mathbf{e}_j
      + \frac{T\adv-z}{T\mathbf{e}_k}\resnv_k\mathbf{e}_k + \adv \right)\\
      &= T\left(\frac{z - T\adv}{T\mathbf{e}_j}\mathbf{e}_j + \adv \right)\\
      &= \frac{z - T\adv}{T\mathbf{e}_j}T\mathbf{e}_j + T\adv\\
      &= z
    \end{align*}
  \end{minipage}
  \end{center}
\end{proof}

\lemimage*
\begin{proof}
  Suppose $\vec{u} \rightarrow_{\VASSYM{}} \vec{v}$.  Then there
  exists a transformer $(\res, \adv) \in \VASSYM{}$ such that
  $\vec{v} = \res * \vec{u} + \adv$.  It follows that $T\vec{v} = T(\res \hadamard{} \vec{u} + \adv)
  = T(\res \hadamard{} \vec{u}) + T\adv = (\rim{T}{\res}) \hadamard{}
  (T\vec{u}) + T \adv$.  Since $(\rim{T}{\res}, T\adv) \in \image{\VASSYM{}}{T}$, $T\vec{u}
   \rightarrow_{\image{\VASSYM{}}{T}} T\vec{v}$.  The other direction is symmetric.
\end{proof}

\join*
\begin{proof}
Algorithm~\ref{alg:join} takes as input $(S^1, \VASSYM{}^1)$ 
and $(S^2, \VASSYM{}^2)$ and constructs matrices $T^1$ and $T^2$
such that $\VASSYM{}^1 \models_{T^1} \VASSYM{}$, $\VASSYM{}^2 \models_{T^2} \VASSYM{}$,
and $T^1 S^1 = S = T^2 S^2$. Clearly, $(S, \VASSYM{})$ is an upper bound of $(S^1, \VASSYM{}^1)$ 
and $(S^2, \VASSYM{}^2)$ in the $\preceq$ order. We proceed by showing that $(S, \VASSYM{})$ is
a least upper bound of $(S^1, \VASSYM{}^1)$ and $(S^2, \VASSYM{}^2)$
(during which we also show $(S, \VASSYM{})$ is normal).
 
Let $(\widetilde{S}, \widetilde{\VASSYM{}})$ be a \VAS{} abstraction that is an upper bound of
 $(S^1, \VASSYM{}^1)$ and $(S^2, \VASSYM{}^2)$ in the $\preceq$ order.
By definition of $\preceq$, there exist linear transformations $\widetilde{T}^1$ and $\widetilde{T}^2$
such that $\widetilde{T}^1 S^1 = \widetilde{S} = \widetilde{T}^2 S^2$. Furthermore, by Lemma~\ref{lem:coherence},
$\widetilde{T}^1$ must be coherent with respect to $\VASSYM{}^1$ and 
$\widetilde{T}^2$ must be coherent with respect to $\VASSYM{}^2$.
To prove that $(S,\VASSYM{})$ is a \textit{least} upper bound, we need to show that
$(S, \VASSYM{}) \preceq (\widetilde{S}, \widetilde{\VASSYM{}})$. Recall that 
$(S, \VASSYM{}) \preceq (\widetilde{S}, \widetilde{\VASSYM{}})$ is defined by existence of a linear transformation $T$ such
that (1) $TS = \widetilde{S}$ and (2) for all $\vec{u} \rightarrow_{\VASSYM{}} \vec{v}$, we have that
$T \vec{u} \rightarrow_{\widetilde{\VASSYM{}}} T \vec{v}$. We show that 
$(S, \VASSYM{}) \preceq (\widetilde{S}, \widetilde{\VASSYM{}})$ by constructing a matrix $T$
such that $T T^1 = \widetilde{T}^1$ and $T T^2 = \widetilde{T}^2$.

Let us first reason about
why constructing a matrix $T$ such that $T T^1 = \widetilde{T}^1$ and $T T^2 = \widetilde{T}^2$ 
is sufficient to prove that $(S, \VASSYM{}) \preceq (\widetilde{S}, \widetilde{\VASSYM{}})$.
First, recall that $T^1 S^1 = S$ and that $\widetilde{T}^1 S^1 = \widetilde{S}$.
Thus, $T S = T T^1 S^1$ and by substituting $T T^1$ with $\widetilde{T}^1$, we arrive at
$T S = \widetilde{T}^1 S^1 =  \widetilde{S}$.
Next, we show that if $\vec{u} \rightarrow_{\VASSYM{}} \vec{v}$ then $T \vec{u} \rightarrow_{\widetilde{\VASSYM{}}} T \vec{v}$.
Observe that $\VASSYM{}$ is constructed as the union of the image of $\VASSYM{}^1$ under $T^1$ together with the
image of $\VASSYM{}^2$ under $T^2$. Lemma~\ref{lem:image} informs us that if $\vec{u} \rightarrow_{\VASSYM{}} \vec{v}$,
then there must exist a $\bar{\vec{u}}$ and a $\bar{\vec{v}}$ such that either $T^1 \bar{\vec{u}} = \vec{u}$,
$T^1 \bar{\vec{v}} = \vec{v}$, and $\bar{\vec{u}} \rightarrow_{\VASSYM{}^1} \bar{\vec{v}}$
(note also that in this case $\vec{u} \rightarrow_{\VASSYM{}} \vec{v}$ can be equivalently expressed 
$T^1 \bar{\vec{u}} \rightarrow_{\VASSYM{}} T^1 \bar{\vec{v}}$), or
$T^2 \bar{\vec{u}} = \vec{u}$, $T^2 \bar{\vec{v}} = \vec{v}$, and $\bar{\vec{u}} \rightarrow_{\VASSYM{}^2} \bar{\vec{v}}$
(note also that in this case $\vec{u} \rightarrow_{\VASSYM{}} \vec{v}$ can be equivalently expressed
$T^2 \bar{\vec{u}} \rightarrow_{\VASSYM{}} T^2 \bar{\vec{v}}$).
Since $(\widetilde{S}, \widetilde{\VASSYM{}})$ is an upper bound of $(S^1, \VASSYM{}^1)$ and $(S^2, \VASSYM{}^2)$,
the former case implies that
$\widetilde{T}^1 \bar{\vec{u}} \rightarrow_{\widetilde{\VASSYM{}}} \widetilde{T}^1 \bar{\vec{v}}$
(and by substitution with $T T^1 = \widetilde{T}^1$, we have that
$T T^1 \bar{\vec{u}} \rightarrow_{\widetilde{\VASSYM{}}} T T^1 \bar{\vec{v}}$)
and the latter case implies that
$\widetilde{T}^2 \bar{\vec{u}} \rightarrow_{\widetilde{\VASSYM{}}} \widetilde{T}^2 \bar{\vec{v}}$
(and by substitution with $T T^2 = \widetilde{T}^2$, we have that
$T T^2 \bar{\vec{u}} \rightarrow_{\widetilde{\VASSYM{}}} T T^2 \bar{\vec{v}}$). Thus,
$T \vec{u} \rightarrow_{\widetilde{\VASSYM{}}} T \vec{v}$.

We now show how to construct a matrix $T$ such that 
$T T^1 = \widetilde{T}^1$ and $T T^2 = \widetilde{T}^2$. 
We construct $T$ on a row by row level, showing that
for each row $i$ of $\widetilde{T}^1$ ($\widetilde{T}^2$ is the same size),
 there is a vector $\vec{t}^i$ such that
$\vec{t}^i T^1 = \widetilde{T}^1_i$ and $\vec{t}^i T^2 = \widetilde{T}^2_i$.
$T_i$ is then simply equal to $\vec{t}^i$.

We proceed by reasoning about coherence classes. Recall that a linear simulation
of a \VAS{} must be coherent with respect to that \VAS{} (if the result of the simulation is a new \VAS{}).
So, $\widetilde{T}^1_i$ is a row vector that is coherent with respect to $\VASSYM{}^1$
and $\widetilde{T}^2_i$ is a row vector that is coherent with respect to $\VASSYM{}^2$. 
Thus, there is a coherence
class $C^{i,1}$ of $\VASSYM{}^1$ and a coherence class $C^{i,2}$ of
$\VASSYM{}^2$ such that $\widetilde{T}^1_i = \widetilde{\vec{t}}^{i,1} \proj{C^{i,1}}$ and
$\widetilde{T}^2_i = \widetilde{\vec{t}}^{i,2} \proj{C^{i,2}}$ for some vectors $\widetilde{\vec{t}}^{i,1}$
and $\widetilde{\vec{t}}^{i,2}$. Observe that 
\[\mathcal{T}(C^{i,1},C^{i,1}) \defeq \{ (\widetilde{\vec{t}}^{i,1}\proj{C^{i,1}}, \widetilde{\vec{t}}^{i,2}\proj{C^{i,2}}) : \widetilde{\vec{t}}^{i,1}\proj{C^{i,1}}S^1 = 
\widetilde{\vec{t}}^{i,2}\proj{C^{i,2}}S^2 \}\]
is a vector space. 
For each coherence class $C^{i,1}$ of $\VASSYM{}^1$ and $C^{i,2}$ of $\VASSYM{}^2$, there is a set $C^{(i,1),(i,2)}$
such that (1) the rows of $\proj{C^{(i,1),(i,2)}}S$ forms a basis for the intersection of the rowspace
of $\proj{C^{i,1}}S^1$ with the rowspace of $\proj{C^{i,2}}S^2$ 
(this follows directly from the $\textit{pushout}$ procedure used in the algorithm)
and (2) $C^{(i,1),(i,2)}$ is a coherence class of $\VASSYM{}$. To see that $C^{(i,1),(i,2)}$ is a coherence class
of $\VASSYM{}$, observe that (1) for a \VAS{} $\bar{\VASSYM{}}$ constructed as
the image of another \VAS{} $\hat{\VASSYM{}}$ under a coherent linear transformation $\hat{T}$ with no zero rows,
$j \equiv_{\bar{\VASSYM{}}} k$ if and only if $\hat{T_j}$ and $\hat{T_k}$ both act on 
(contain non-zero values exclusively in columns corresponding to)
the same coherence class of $\hat{\VASSYM{}}$; and (2)
the coherence classes of a \VAS{} constructed as 
the union of two \VAS{}s $\hat{\VASSYM{}}$ and $\bar{\VASSYM{}}$ 
is equal to the pairwise intersection of each
coherence class of $\hat{\VASSYM{}}$ with each coherence class of $\bar{\VASSYM{}}$.

Note then that (1) for each coherence class $C$ of $\VASSYM{}$, the rows of $\proj{C} S$ form 
a basis for a vector space and thus $(S, \VASSYM{})$ is normal; and (2) there exists a $\vec{t}^i$ and a $\bar{\vec{t}}^i$
such that 
$\vec{t}^i S = \vec{t}^i T^1 S^1 = \widetilde{T}^1_i S^1$,
$\vec{t}^i S = \vec{t}^i T^2 S^2 = \widetilde{T}^2_i S^2$, and 
$\vec{t}^i =  \bar{\vec{t}}^i \proj{C^{(i,1),(i,2)}}$.
This implies that there exists $\vec{t}^{i,1}$ and $\vec{t}^{i,2}$ 
such that $\vec{t}^i T^1 = \vec{t}^{i,1} \proj{C^{i,1}}$ and 
$\vec{t}^i T^2 = \vec{t}^{i,2} \proj{C^{i,2}}$. Both $\proj{C^1} S^1$ and $\proj{C^2} S^2$ are invertible and thus we use these
equations together with the fact that $\widetilde{T}^1_i = \widetilde{\vec{t}}^{i,1} \proj{C^{i,1}}$ and
$\widetilde{T}^2_i = \widetilde{\vec{t}}^{i,2} \proj{C^{i,2}}$ to arrive at
$\vec{t}^i T^1 = \widetilde{T}^1_i$ and $\vec{t}^i T^2 = \widetilde{T}^2_i$.

\end{proof}

\vasabstraction*
\begin{proof}
We break this proof into three steps. We first show that the output of Algorithm~\ref{alg:core}, $(S,\VASSYM{})$, 
is a \VAS{} abstraction of its input transition formula $F$ at termination.
We next show that Algorithm~\ref{alg:core} always terminates. And we finally show that 
$(S,\VASSYM{})$ is a best abstraction of $F$ when $F$ is in $\LRA$.

\begin{enumerate}
\item
Supposing that Algorithm~\ref{alg:core} terminates, its output is a \VAS{} abstraction of $F$.
Recall that 
Algorithm~\ref{alg:core} terminates when $F \wedge \lnot \gamma(S,\VASSYM{})$ is unsatisfiable.

We prove this by contradiction. 
Assume that the algorithm terminated
and that there exists a $\vec{u}$ and $\vec{v}$ such that $\vec{u} \rightarrow_F \vec{v}$ and $S \vec{u} \not\rightarrow_{\VASSYM{}} S \vec{v}$. Then $\vec{u} \rightarrow_{F \wedge \lnot \gamma(S,\VASSYM{})} \vec{v}$; the algorithm would not have terminated yet.

\item
Next, we prove Algorithm~\ref{alg:core} always terminates. 
Let $(S^k,\VASSYM{}^k)$ denote the \VAS{} abstraction obtained 
just before Algorithm~\ref{alg:core} enters its $k^{th}$ loop iteration and let
$\Gamma^k$ denote the formula $\Gamma$ obtained at the same point in time.
Recall that $\Gamma$ is the transition formula
whose models are the transitions of $F$ that are not simulated by $(S,\VASSYM{})$.
Algorithm~\ref{alg:core} will only enter the $k^{th}$ loop iteration
if $\Gamma^k$ is satisfiable. 
If $\Gamma^k$ is satisfiable,
then there must exist some model $M$ of $\Gamma^k$ and 
some cube $C$ of the DNF of $F$ such that $M \models C$.
Algorithm~\ref{alg:core} sets
$(S^{k+1},\VASSYM{}^{k+1})$ equal to $(S^k, \VASSYM{}^k) \sqcup \hat{\alpha}(C)$.
By Proposition~\ref{prop:alphahat}, $\hat{\alpha}(C)$ is a \VAS{} abstraction of $C$.
For any formula $C$, a \VAS{} abstraction that is a least upper bound of a \VAS{} abstraction of $C$ 
together with any other \VAS{} abstraction must also be a \VAS{} abstraction of $C$. 
Thus, there will never be a model $M$ of $\Gamma^{k'}$ for any $k' > k$
such that $M \models C$.
Transition formulas are finite in size and have a finite number of cubes. Since a new
cube must be witnessed on each iteration, Algorithm~\ref{alg:core} must terminate.

\item
We now prove that $(S,\VASSYM{})$ is a best abstraction of $F$ when $F$ is in $\LRA$.
Assume that $F$ is in $\LRA$ and
let $(\widetilde{S}, \widetilde{\VASSYM{}})$ be a \VAS{} abstraction of $F$. We show that 
$(S,\VASSYM{}) \preceq (\widetilde{S},\widetilde{\VASSYM{}})$ via induction on the number of loop iterations in Algorithm~\ref{alg:core}.

Initially, $(S,\VASSYM{}) = (I, \emptyset)$. $(I, \emptyset)$ is less in the $\preceq$ order than any \VAS{} abstraction of the
same concrete dimension. So clearly $(I, \emptyset) \preceq (\widetilde{S},\widetilde{\VASSYM{}})$.

Let $(S^k,\VASSYM{}^k)$ denote the \VAS{} abstraction obtained
just before Algorithm~\ref{alg:core} enters its $k^{th}$ loop iteration.
Assume as the induction hypothesis that $(S^k,\VASSYM{}^k) \preceq (\widetilde{S},\widetilde{\VASSYM{}})$.
Let the $C$ denote the $\LRA$ cube of the DNF of $F$ selected on the $k^{th}$ iteration of the loop.
Then $(S^{k+1},\VASSYM{}^{k+1}) = (S^k, \VASSYM{}^k) \sqcup \hat{\alpha}(C)$. Since $C \models F$ and $F \Vdash_{\widetilde{S}} \widetilde{\VASSYM{}}$, it must be the case that $C \Vdash_{\widetilde{S}} \widetilde{\VASSYM{}}$. By Proposition~\ref{prop:alphahat},
$\hat{\alpha}(C)$ is a \VAS{} best abstraction for $C$ when $C$ is in $\LRA$.
Thus,
$\hat{\alpha}(C) \preceq (\widetilde{S},\widetilde{\VASSYM{}})$.
By induction hypothesis, $(S^k,\VASSYM{}^k) \preceq (\widetilde{S},\widetilde{\VASSYM{}})$. 
Therefore, by Proposition~\ref{prop:join}, $(S^{k+1},\VASSYM{}^{k+1}) \preceq (\widetilde{S},\widetilde{\VASSYM{}})$. Thus, $(S,\VASSYM{})$ is a best abstraction of $F$.

\end{enumerate}
\end{proof}

\vass*
\begin{proof}
Let $(S, \SYMVASS{}) = (S, (P, E))$ be the \VASS{} abstraction computed by Algorithm~\ref{alg:vasscore}
with input transition formula $F$ and input control states $P$. We first show
that $(S, \SYMVASS{})$ is a \VASS{} abstraction of $F$.

Define $\mathcal{E}^{p, q}(E) = \{(\res, \adv) : (p, (\res, \adv), q) \in E\}$.  
Intuitively, $(S, \mathcal{E}^{p, q}(E))$ is the \VAS{} abstraction derived from the
 transformers that go from control state $p$ to control state $q$ in the \VASS{} abstraction $(S, (P, E))$.
 Observe that for any \VASS{} abstraction $(\widetilde{S}, (P, \widetilde{E}))$ of $F$ 
 with control states $P$, we have that $F \Vdash_{\widetilde{S}} (P, \widetilde{E})$ 
 when (exactly when in $\LRA$ case) for each pair $p, q \in P$, there exists a linear
 simulation $\widetilde{T}^{p,q}$ such that 
 $\overline{\VASSYM{}}^{p, q} \Vdash_{\widetilde{T}^{p,q}} \mathcal{E}^{p, q}(\widetilde{E})$
 and $\widetilde{T}^{p,q} \overline{S}^{p, q} = \widetilde{S}$, where 
 $(\overline{S}^{p, q}, \overline{\VASSYM{}}^{p, q}) =  \texttt{abstract-VASR}(p(\vec{x})  \land F(\vec{x},\vec{x}') \land q(\vec{x}'))$. Algorithm~\ref{alg:vasscore} constructs $(S, \SYMVASS{})$ by letting
 $S$ equal the same simulation matrix in the \VAS{} abstraction 
$\bigsqcup_{p, q \in P} (\overline{S}^{p, q}, \overline{\VASSYM{}}^{p, q})$ and letting 
$\mathcal{E}^{p, q}(E)$ be the image of $\overline{\VASSYM{}}^{p, q}$ under the simulation matrix
$T^{p,q}$ (here a simulation from $(\overline{S}^{p, q}, \overline{\VASSYM{}}^{p, q})$ to 
$\bigsqcup_{p, q \in P} (\overline{S}^{p, q}, \overline{\VASSYM{}}^{p, q})$).
So clearly $(S, \SYMVASS{})$ is a \VASS{} abstraction of $F$.

We proceed by showing that if $F$ is in $\LRA$, then for any other \VASS{}
abstraction $(\widetilde{S}, \widetilde{\SYMVASS{}}) = (\widetilde{S}, (P, \widetilde{E}))$
of $F$ (with the same set of control states $P$), $(S, (P, E)) \preceq (\widetilde{S}, (P, \widetilde{E}))$.
It is sufficient to prove that there exists a $T$ such that $T S = \widetilde{T}$ and for each pair $p, q \in P$,
$\mathcal{E}^{p, q}(E) \Vdash_T \mathcal{E}^{p, q}(\widetilde{E})$.

From this point forward this proof is fairly similar the proof of Proposition~\ref{prop:join}.
Let $T^{p,q}$ be a simulation from $(\overline{S}^{p, q}, \overline{\VASSYM{}}^{p, q})$
to $(S, \mathcal{E}^{p, q}(E))$ and let $\widetilde{T}^{p,q}$ be a simulation from
$(\overline{S}^{p, q}, \overline{\VASSYM{}}^{p, q})$ to $(\widetilde{S}, \mathcal{E}^{p, q}(\widetilde{E}))$
(recall that $\widetilde{T}^{p,q}$ must exist since $F$ is $\LRA{}$).
We construct $T$ such that for any pair $p, q \in P$, we have that $T T^{p,q} = \widetilde{T}^{p,q}$.
Observe that $\mathcal{E}^{p, q}(E) \Vdash_T \mathcal{E}^{p, q}(\widetilde{E})$ and $T S = \widetilde{S}$
naturally follow
from such a construction:
If $\vec{u} \rightarrow_{\mathcal{E}^{p, q}(E)} \vec{v}$, then by Lemma~\ref{lem:image}
there is some $\bar{\vec{u}}$ and some $\bar{\vec{v}}$ such that $T^{p,q} \bar{\vec{u}} = \vec{u}$,
$T^{p,q} \bar{\vec{v}} = \vec{v}$, and
$\bar{\vec{u}} \rightarrow_{\overline{\VASSYM{}}^{p, q}} \bar{\vec{v}}$.
Then, by definition of linear simulation, we must have
$\widetilde{T}^{p,q} \bar{\vec{u}} \rightarrow_{\mathcal{E}^{p, q}(\widetilde{E})} \widetilde{T}^{p,q} \bar{\vec{u}}$.
From here, substitution gives us to 
$T \vec{u} \rightarrow_{\mathcal{E}^{p, q}(\widetilde{E})} T \vec{u}$.
Substitution also gives us 
$T S = T T^{p,q} \overline{S}^{p, q} = \widetilde{T}^{p,q} \overline{S}^{p, q} = \widetilde{S}$.

We are ultimately just reasoning about matrix multiplication and so we can
reason row by row. We show that for each row $i$ of $\widetilde{T}^{p,q}$ there exists a vector $\vec{t}^i$
such that for all $p, q \in P$ we have that $\vec{t}^i T^{p,q} = \widetilde{T}^{p,q}_i$. We prove that
such a $\vec{t}^i$ exists by simultaneously reason about coherence classes of the \VAS{}es 
$\overline{\VASSYM{}}^{p, q}$ for all $p, q \in P$.
For notational simplicity, we henceforth write
a pair $p, q \in P$ as an element of the cartesian product of $P$ with itself, $P \times P$,
and we let $|P \times P| = n$.

For any $j \in P \times P$, we have that $\widetilde{T}^{j} \overline{S}^{j} = \widetilde{S}$.
Thus, for a fixed row $i$ of $\widetilde{S}$, $\widetilde{S}_i$, we have that
$\widetilde{T}^{j}_i \overline{S}^{j} = \widetilde{S}_i$. 
We can rewrite $\widetilde{T}^{j}_i \overline{S}^{j}$ as $\widetilde{\vec{t}}^{i,j} \proj{C^{i,j}} \overline{S}^{j}$
for some coherence class $C^{i,j}$ of $\overline{\VASSYM{}}^{j}$ and some vector $\widetilde{\vec{t}}^{i,j}$.
Observe that
\[\mathcal{T}(C^{i,1},...C^{i,n}) \defeq 
\{ (\widetilde{\vec{t}}^{i,1}\proj{C^{i,1}},..., \widetilde{\vec{t}}^{i,n}\proj{C^{i,n}}) : \widetilde{\vec{t}}^{i,1}\proj{C^{i,1}}S^1 = ... =
\widetilde{\vec{t}}^{i,n}\proj{C^{i,n}}S^n \}\]
is a vector space.

So there exists a set $C^{(i,1),...,(i,n)}$ such that (1)
$\proj{C^{(i,1),...,(i,n)}} S$ is a basis for ($\cap_{j \in P \times P}$ the rowspace of $\proj{C^{i,j}}S^j$) 
and (2) $C^{(i,1),...,(i,n)}$ is a coherence class of
$(S, \SYMVASS{})$. To witness this fact, recall that
the simulation matrix of $(S, \SYMVASS{})$ is equal to the simulation matrix of
$\bigsqcup_{j \in P \times P} (\overline{S}^{j}, \overline{\VASSYM{}}^{j})$ and observe that the coherence classes
of $(S, \SYMVASS{})$ are equal to the coherence classes of $\bigsqcup_{j \in P \times P} (\overline{S}^{j}, \overline{\VASSYM{}}^{j})$.
See the proof of Proposition~\ref{prop:join} for a better understanding of why a basis for 
($\cap_{j \in P \times P}$ the rowspace of $\proj{C^{i,j}}S^j$) forms rows of $S$ belonging to the same coherence class of
$\bigsqcup_{j \in P \times P} (\overline{S}^{j}, \overline{\VASSYM{}}^{j})$ .

Thus, there exist some vector $\vec{t}^i$ and some vector $\bar{\vec{t}}^i$
 such that, $\vec{t}^i = \bar{\vec{t}}^i \proj{C^{(i,1),...,(i,n)}}$ and such that
 for all $j \in P \times P$, we have 
 $\vec{t}^i S = \widetilde{T}^j_i \overline{S}^{j}$.
 The first statement taken together with the fact that $\vec{t}^i S =  \vec{t}^i T^j \overline{S}^{j}$
 for all $j \in P \times P$
 implies that there exists a $\bar{\vec{t}}^{i,j}$
 such that $\vec{t}^i T^j = \bar{\vec{t}}^{i,j} \proj{C^{i,j}}$. Recall that for all $j \in P \times P$,
 $\proj{C^{i,j}} \overline{S}^{j}$ is invertible (\texttt{abstract-VASR} produces normal \VAS{}). 
 Recall also that we can rewrite $\widetilde{T}^j_i$ as $\widetilde{\vec{t}}^{i,j} \proj{C^{i,j}}$.
 So $\bar{\vec{t}}^{i,j} \proj{C^{i,j}} \overline{S}^{j} = \widetilde{\vec{t}}^{i,j} \proj{C^{i,j}} \overline{S}^{j}$
 and from this we can deduce that $\vec{t}^i T^j = \widetilde{T}^j_i$.
  \end{proof}


\end{document}